\documentclass[conference]{IEEEtran}
\IEEEoverridecommandlockouts
\usepackage{cite}
\usepackage{amsmath,amssymb,amsfonts}
\usepackage{algorithmic}
\usepackage{graphicx}
\usepackage{textcomp}
\usepackage{xcolor}

\usepackage{ntheorem}
\usepackage{makecell}
\usepackage{booktabs}
\usepackage{multirow}
\usepackage{algorithm}
\usepackage{subfigure}
\usepackage{makecell}
\usepackage{url}

\newtheorem{theorem}{Theorem}
\newtheorem*{proof}{Proof}
\newtheorem{lemma}{Lemma}

\DeclareMathOperator*{\argmax}{argmax}

\def\BibTeX{{\rm B\kern-.05em{\sc i\kern-.025em b}\kern-.08em
    T\kern-.1667em\lower.7ex\hbox{E}\kern-.125emX}}
\begin{document}

\title{MLANE: Meta-Learning Based Adaptive Network Embedding\\
}

\author{\IEEEauthorblockN{Chen Cui}
\IEEEauthorblockA{\textit{School of Computer Science} \\
\textit{Sichuan University}\\
Chengdu, China \\
cuichen@stu.scu.edu.cn}
\and
\IEEEauthorblockN{Ning Yang\thanks{\IEEEauthorrefmark{1} Ning Yang is the corresponding author.}\IEEEauthorrefmark{1} }
\IEEEauthorblockA{\textit{School of Computer Science} \\
\textit{Sichuan University}\\
Chengdu, China \\
yangning@scu.edu.cn}
\and
\IEEEauthorblockN{Philip S. Yu}
\IEEEauthorblockA{\textit{Department of Computer Science} \\
\textit{University of Illinois at Chicago}\\
Chicago, USA \\
psyu@uic.edu}
}

\maketitle

\begin{abstract}
Most existing random walk based network embedding methods often follow only one of two principles, homophily or structural equivalence. 
In real world networks, however, nodes exhibit a mixture of homophily and structural equivalence, which requires \textit{adaptive network embedding} that can adaptively preserve both homophily and structural equivalence for different nodes in different down-stream analysis tasks. In this paper, we propose a novel method called Meta-Learning based Adaptive Network Embedding (MLANE), which can learn adaptive sampling strategy for different nodes in different tasks by incorporating sampling strategy learning with embedding learning into one optimization problem that can be solved via an end-to-end meta-learning framework. In extensive experiments on real datasets, MLANE shows significant performance improvements over the baselines. The source code of MLANE and the datasets used in experiments and all the hyperparameter settings for baselines are available at https://github.com/7733com/MLANE.
\end{abstract}


\section{Introduction}
\label{Sec:Intro}
Network embedding, which aims at learning low-dimensional vectorial feature representations for nodes in a network and preserving structural properties of nodes, has been attracting increasing interest from the research community in recent years due to its promising performance in various network analysis tasks \cite{cui2018survey}. The existing random walk based network embedding methods often follow only one of two principles, homophily or structural equivalence \cite{mcpherson2001birds,hoff2002latent}. According to homophily principle, embeddings of the nodes that are interconnected within the same community (e.g., purple nodes in Figure \ref{Fig:Example}) should be closer than the embeddings of the nodes that belong to different communities \cite{Zhu8329541}. In contrast, according to structural equivalence principle the nodes with the same structural role will have closer embeddings, even though they are far away from each other (e.g., two red nodes in Figure \ref{Fig:Example}) \cite{grover2016node2vec}.

%In terms of what kind of structural property is preserved, the existing network embedding methods can be roughly categorized into two classes. One class of them is to preserve \textit{homophily} of nodes \cite{mcpherson2001birds,perozzi2014deepwalk,Zhu8329541}, while the other class is to preserve \textit{structural equivalence} of nodes \cite{lorrain1971structural,grover2016node2vec}. Preserving homophily regularizes the learned embeddings with local connectivity so that the interconnected nodes have similar representations. The other class is to preserve structural equivalence
%For example, DeepWalk\cite{perozzi2014deepwalk} and DHPE\cite{Zhu8329541} learn the node embeddings by preserving the first-order proximity and high-order proximity between nodes, respectively. 
%Preserving homophily will benefit tasks like label propagation and community detection, as nodes with similar label or features are more likely to interconnect, but often fail in tasks like structure role identity \cite{Rossi6880836,ribeiro2017struc2vec}. In structural role identity, the nodes with similar local topology, even though without connection, play the same function and will be identified with the same role, where structural equivalence between nodes is the property desired to be preserved in the embedding learning. %For example, by capturing structural equivalence between nodes, struct2vec\cite{ribeiro2017struc2vec} can generate similar embeddings for nodes of similar roles while dissimilar embeddings for nodes of different roles. 

In real world networks, however, nodes usually exhibit a mixture of homophily and structural equivalence, which requires an \textit{adaptive network embedding} framework that can adaptively preserve both homophily and structural equivalence for different nodes in different down-stream analysis tasks \cite{grover2016node2vec}. The problem of adaptive network embedding is not easy due to the following three challenges. 

%For example, Figure \ref{Fig:Example} shows a network in recommender system, where nodes $v_1$ and $v_2$ are items, the other nodes are users, and an edge represent an interaction between a user and an item. Suppose we want to classify the nodes into three classes, the item class $\{ v_1,v_2\}$, the user class that interacts with $v_1$, and the user class that interacts with $v_2$, just based on their embeddings preserving their structural properties. For this purpose, the random walks for $v_1$ and $v_2$ should choose sampling strategy BFS with larger probability, which is in favor of the exploration of local neighborhoods, so that the embeddings of $v_1$ and $v_2$ can be biased to preserving their structural equivalence which is ascertained by their neighborhoods. On the contrary, to identify the user nodes that interacts with the same item, say $u_1$ and $u_2$, the random walks for them should prefer to sampling strategy DFS, which is in favor of determining the boundary of the community that $u_1$ or $u_2$ belong to, so that their embeddings can be biased to preserving their homophily. 

  \begin{figure}[t]
    \centering
    \includegraphics[width=0.7\linewidth]{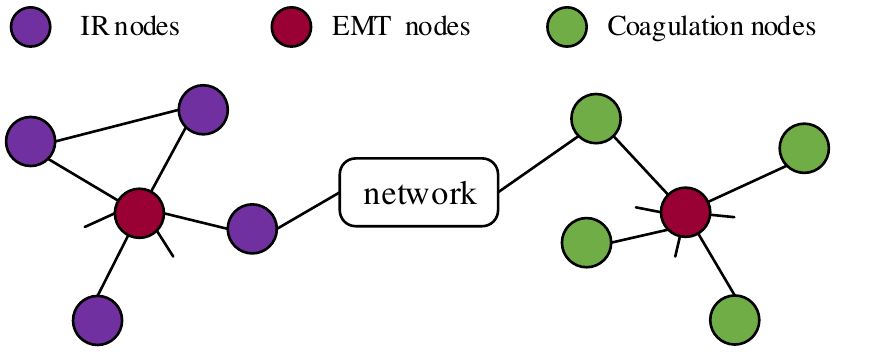}
    \caption{Illustration of adaptive embedding.}
    \label{Fig:Example}
  \end{figure}

\begin {itemize}

\item \textbf{Node-Adaptive Sampling} In real world networks, homophily and structural equivalence likely make different contribution to the embeddings of different nodes, which suggests that the sampling process should pay more attention to homophily for some nodes, while to structural equivalence for other nodes. For example, Figure \ref{Fig:Example} shows a piece of Protein-Protein Interaction (PPI) network \cite{stark2010biogrid}, where a node represents a Homo Sapiens gene, and an edge represents an physical interaction between two genes. In Figure \ref{Fig:Example}, there are three types of nodes, where the red nodes are the genes related to Epithelial-Mesenchymal Transition (EMT), the purple nodes are the genes related to Inflammatory Response (IR), and the green nodes are the genes related to Coagulation. Intuitively, we can see that although the EMT (red) nodes are far away from each other, they are of the same type due to their similar local topology. Furthermore, the IR (purple) nodes and Coagulation (green) nodes are of different types as they are close to different EMT nodes. Therefore, in order to correctly classify these nodes based on their embeddings that preserve their structural properties, the random walks for EMT nodes should choose sampling strategy BFS (Breadth First Search) with larger probability, which is in favor of exploring local neighborhoods, so that the embeddings of EMT nodes can be biased to preserving their structural equivalence which is ascertained by their neighborhoods. On the contrary, the random walks for IR nodes and Coagulation nodes should prefer to sampling strategy DFS (Depth First Search), which is in favor of detecting the nodes that are connected to a same EMT node, so that their embeddings can be biased to preserving their homophily.

%For example, again in Figure \ref{Fig:Example}, nodes $u_1$ and $u_3$ play different roles even though they belong to the same community. The structure of the network will remain almost unchanged if node $u_1$ is removed as it is an edge node, but the network will be broken if node $u_3$ is removed instead as it is a bridge node. In fact, to distinguish their roles, the random walk for edge node $u_1$ should adaptively with a larger probability choose DFS, while for bridge node $u_3$, the random walk should balance DFS with BFS. 

\item \textbf{Task-Adaptive Sampling} In real world, homophily and structural equivalence may have different importance for different tasks. For example, for link prediction, preserving homophily is more important than preserving structural equivalence as links often exist between nodes within the same community, while for tasks like structure role identity \cite{Rossi6880836,ribeiro2017struc2vec}, structural equivalence should contribute more to the embeddings as nodes with similar local topology often play similar role. Therefore, it is desirable that the sampling process can automatically and adaptively assign different weights to homophily and structural equivalence for different network analysis tasks.

\item \textbf{End-to-End Trainability} The existing network embedding methods based on random walk treat the sampling process as a data preprocessing before embedding learning. Separating the sampling process from embedding learning, however, likely leads to suboptimal solutions of the embeddings and the sequent network analysis tasks, as the optimization objective of a separate sampling process may be potentially inconsistent, even conflicting, with that of embedding learning. To avoid this issue, we need to incorporate the sampling process with network embedding so that the sampling strategy and the parameters of the embedding model can be learned together in an end-to-end manner.

\end {itemize}

In this paper, to overcome the above challenges, we propose a novel method called Meta-Learning based Adaptive Network Embedding (MLANE for short). The main idea of MLANE is to make the sampling process learnable in a meta-learning framework so that one node can have its own sampling strategy for its embedding learning to discriminately preserve its homophily and structural equivalence for different tasks. For this purpose, MLANE formulates random walk with a reinforcement learning process, by which the sampling process is parametrized and can be trained to let node and task decide the bias to the two search strategies via BFS and DFS. For the sampling strategy learning, we propose a meta-learning based policy learning algorithm, by which MLANE can incorporate the sampling strategy learning with the embedding learning into one optimization problem that can be solved with an end-to-end optimizing algorithm based on gradient ascent. The contributions of this paper can be summarized as follows:
\begin{itemize}
  \item We propose a novel method called Meta-Learning based Adaptive Network Embedding (MLANE), which can adaptively preserve homophily and structural equivalence for node embeddings by making the random walk based sampling process learnable with a meta-learning framework. To our best knowledge, this is the first time to apply meta-learning to network embedding.
  
  \item We propose a meta-learner for sampling strategy learning, which formulates node sampling as a reinforcement learning process so that the sampling strategy can be learned for different nodes in different tasks with an end-to-end optimizing algorithm.
  
  \item Extensive experiments conducted on real world networks across different domains demonstrate the effectiveness of MLANE. Specifically, the results show that the embeddings learned by MLANE can significantly improve the performance of node classification and link prediction.
\end{itemize}

The rest of this paper is organized as follows. In Section \ref{Sec:Prelim}, we introduce the preliminaries and formally define the target problem of this paper. We present the details of MLANE in Section \ref{Sec:ProposedModel}. In Section \ref{Sec:Experiments}, we empirically evaluate MLANE on node classification, link prediction and node clustering over real world datasets, and verify the adaptiveness of MLANE with case study. At last, we briefly review the related works in Section \ref{Sec:RelatedWork} and conclude in Section \ref{Sec:Conclusion}.

\section{Preliminaries And Problem Definition}
\label{Sec:Prelim}
\subsection{BFS vs DFS}
\begin{figure}[t]
  \centering
  \subfigure[Random walk with BFS]{
    \begin{minipage}[t]{0.5\linewidth}
      \centering
      \label{Fig:BFS}
      \includegraphics[width=\linewidth]{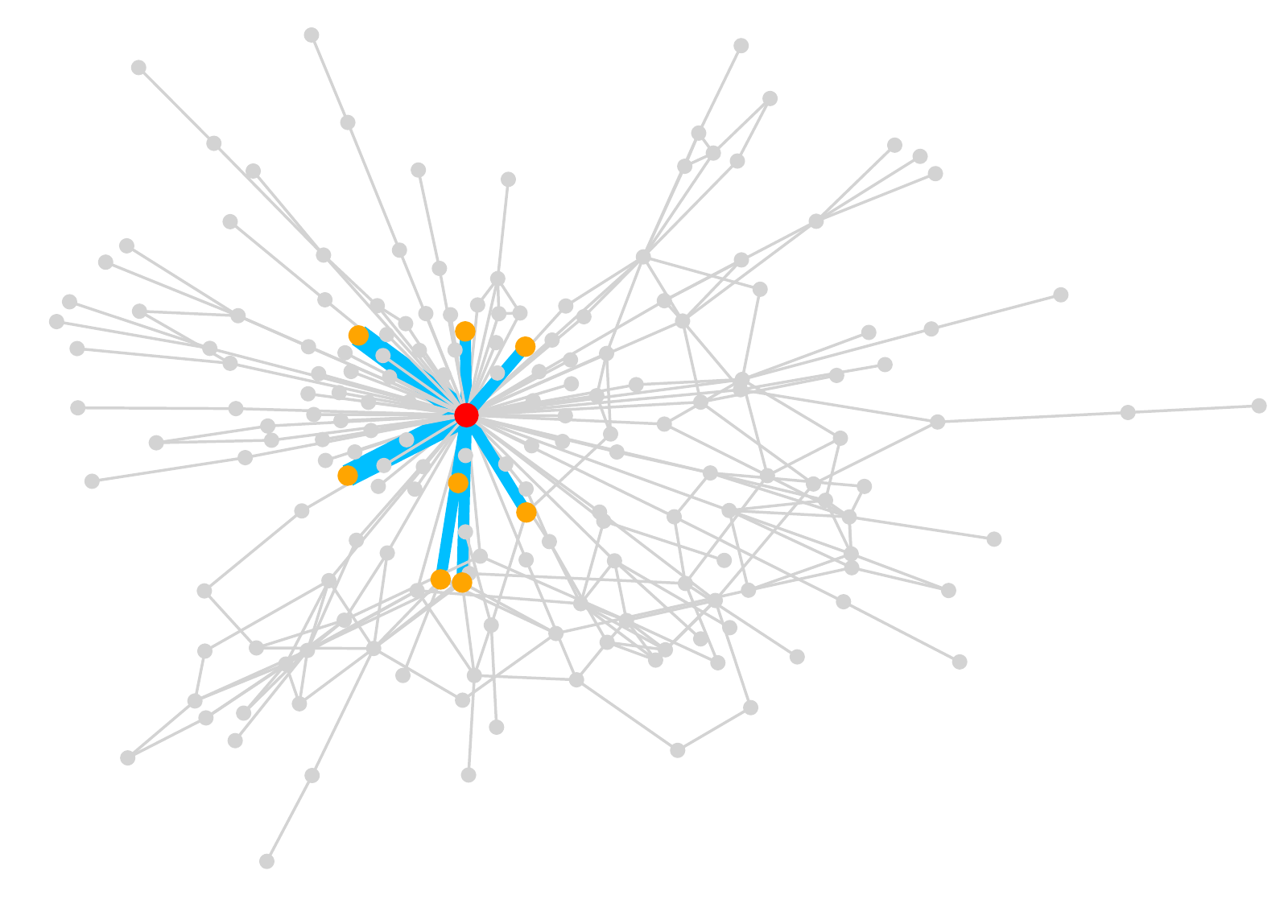}
    \end{minipage}%
  }%
  \subfigure[Random walk with DFS]{
    \begin{minipage}[t]{0.5\linewidth}
      \centering
      \label{Fig:DFS}
      \includegraphics[width=\linewidth]{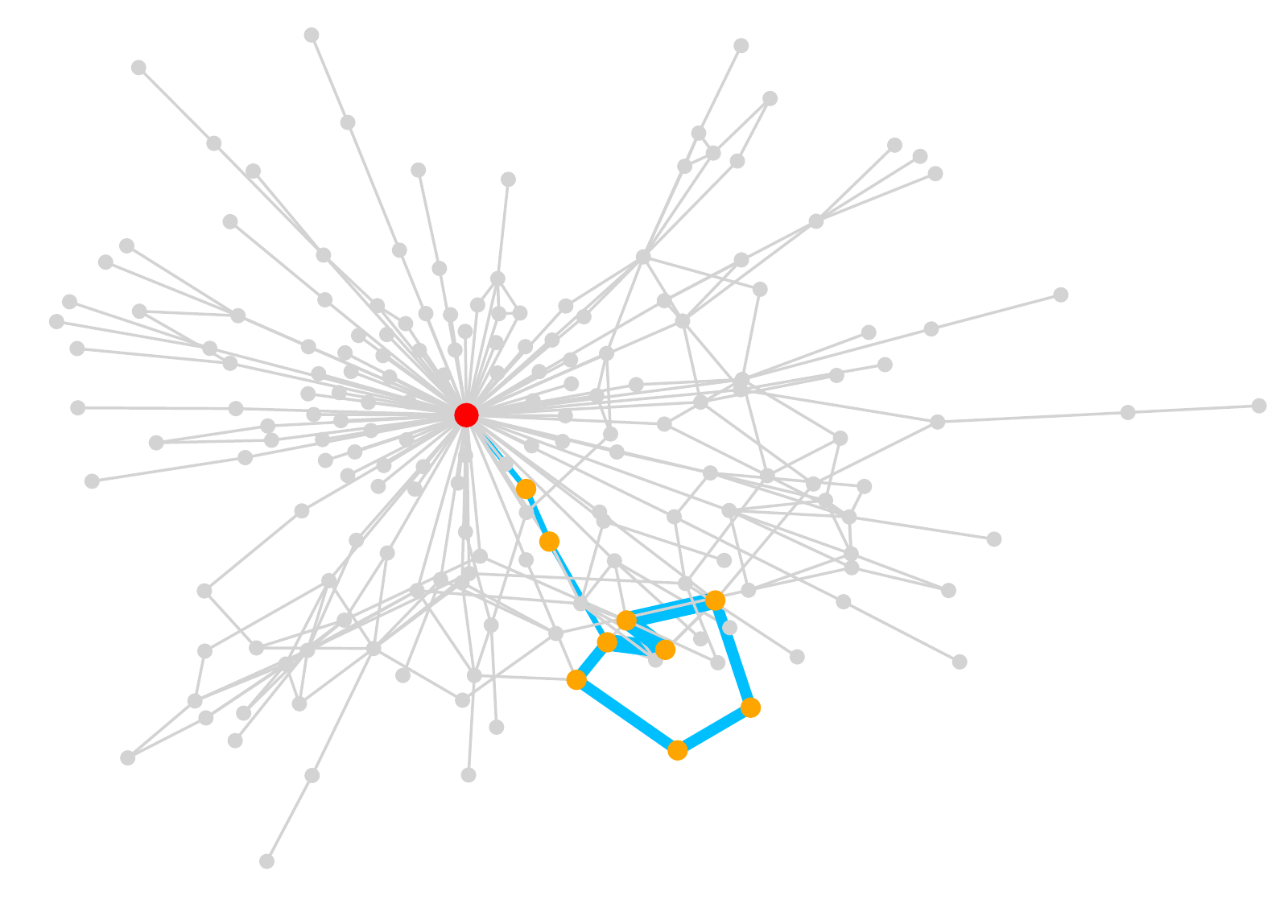}
    \end{minipage}%
  }%
  \caption{Illustration of BFS and DFS.}
  \label{Fig:BFSvsDFS}
\end{figure}
Basically, there are two types of search strategies used in random walks, i.e., BFS (Breadth First Search) and DFS (Depth First Search). Figure~\ref{Fig:BFSvsDFS} shows two random walks respectively based on BFS and DFS, where the red nodes are the source nodes, and the orange nodes are the sampled nodes. The sampling path are represented by the blue edges, and the thicker the blue edges, the more frequent they are visited. As we can see from Figure~\ref{Fig:BFS}, a random walk with BFS strategy tends to sample the nodes close to the source nodes, and some directly connected nodes are even sampled more than once to explore the local topology, which is in favor of preserving structural equivalence of the red node. On the contrary, the random walk shown Figure~\ref{Fig:DFS} is based on DFS strategy, which prefers nodes far from the source. From Figure~\ref{Fig:DFS} we can see that some distant nodes are visited more than once to determine the boundary of a community, which is in favor of preserving homophily of the red node.

\subsection{Problem Definition}
Let $G = (V, E)$ denote a network, where $V$ and $E$ are the node set and edge set, respectively. Suppose we want to learn node embeddings $\boldsymbol{Z} = \{ \boldsymbol{z}_v \in \mathbb{R}^m, v \in V\}$ for a specific network analysis task $T$, where $m$ is the dimensionality of the embeddings. To generate the embedding $\boldsymbol{z}_v$ of a node $v \in V$ for task $T$, a set of node sequences each of which starts from $v$ will be sampled as its context $C_v$, with a learnable policy function $\boldsymbol{\pi}$. Essentially the output of $\boldsymbol{\pi}$ is a probability distribution of search strategies according to which next node can be sampled. %Note that $\boldsymbol{\pi}_{T}(v;$ $\boldsymbol{\theta}^\pi_T)$ is task-specific and dependent on the target node. 

Based on the learnable policy function $\boldsymbol{\pi}$, the random walk can be parametrized as the function $\boldsymbol{S}(v; \boldsymbol{\pi})$, of which the output is just the sampled context, i.e., $C_v = \boldsymbol{S}(v; \boldsymbol{\pi})$. Let $\boldsymbol{C}$ be the contexts of all nodes, i.e., $\boldsymbol{C} = \bigcup_{v \in V}C_v $. Then once the contexts of nodes are sampled, they will be fed into a language model $\boldsymbol{f}$ to generate the node embedding for task $T$, i.e., $ \boldsymbol{Z} = \boldsymbol{f}( \boldsymbol{C})$. Similar to existing works, in this paper we use language model SkipGram \cite{mikolov2013efficient} to generate the node embeddings. At last, the embeddings will be used as input of task $T$ which is evaluated with metric $M_T: \{\boldsymbol{z}_v \in \mathbb{R}^m \} \to \mathbb{R}$. Based on the above definitions, our target problem can be conceptually formulated as follow:

Given an analysis task $T$ over a given network $G = (V, E)$ with evaluation metric $M_T$, we want to learn the policy $\boldsymbol{\pi}$, so that the embeddings $\boldsymbol{Z}$ generated by a given language model $\boldsymbol{f}$ can lead optimal performance of $T$ in terms of metric $M_T$, i.e., 
\begin{equation}
\argmax_{\boldsymbol{\pi}}M_T\big(  \boldsymbol{Z}  \big).
\label{Eq_Problem}
\end{equation}

 \begin{figure}[t]
    \centering
    \includegraphics[width=0.7\linewidth]{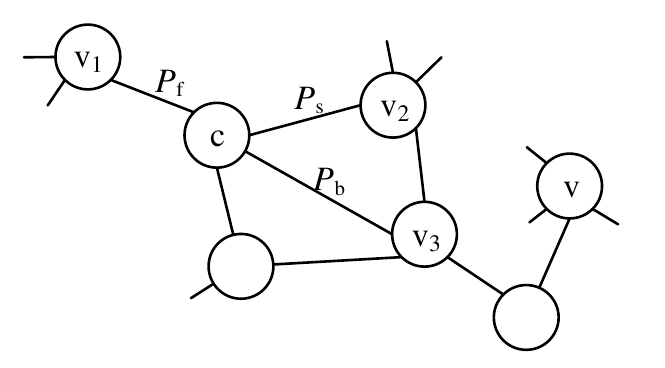}
    \caption{Illustration of sampling.}
    \label{Fig:Sampling}
  \end{figure}

\section{Proposed Model}
\label{Sec:ProposedModel}
\subsection{Sampling Strategy Learning}
As we have mentioned, the proposed model MLANE treats the sampling process $\boldsymbol{S}$ as a Markov Decision Process (MDP) \cite{puterman2014markov} so that the sampling strategy learning can be solved via a meta-learning framework based on reinforcement learning.
 
\subsubsection{State Space}
Intuitively, DFS strategy intends to sample nodes far from source node, while BFS prefers to nodes close to source node. To reflect the effect of search strategy for a given source node, we define the state space of a sampling process combining the source node and the distance from the current sample node to the source node,
\begin{equation}
\mathcal{S} = \{ (v, d)| v \in V, d \in [0,1, \dots, d_{max}] \},
\end{equation}
where $v \in V$ is the source node for which the context is sampled, $d$ is the distance from current sample node to $v$, and $d_{max}$ is the diameter of the network. Note that the initial state of the sampling process for a specific node $v$ is $s_0(v) = (v, 0)$.

\subsubsection{Action Space}
Now the action space consists of three possible actions to search the next node, i.e., 
\begin{equation}
\mathcal{A} = \{ a_\text{f}, a_\text{s}, a_\text{b} \}.
\end{equation} 
Figure \ref{Fig:Sampling} shows an example of random walk starting from node $v$ and now residing node $c$ which is $d_c$ hops apart from $v$. Nodes $v_1$, $v_2$, and $v_3$ are the direct neighbors of node $c$. The first option $a_\text{f}$ is to move one step forward, i.e., sampling $v_1$ as the next sample node. The second option $a_\text{s}$ is to keep the distance unchanged, i.e., sampling $v_2$ as the next sample node. The last option $a_\text{b}$ is to move one step backward, i.e., sampling $v_3$ as the next sample node. 
Note that action $a_\text{f}$ corresponds to DFS, while actions $a_\text{s}$ and $a_\text{b}$ correspond to BFS.

 \begin{figure}[t]
    \centering
    \includegraphics[width=0.8\linewidth]{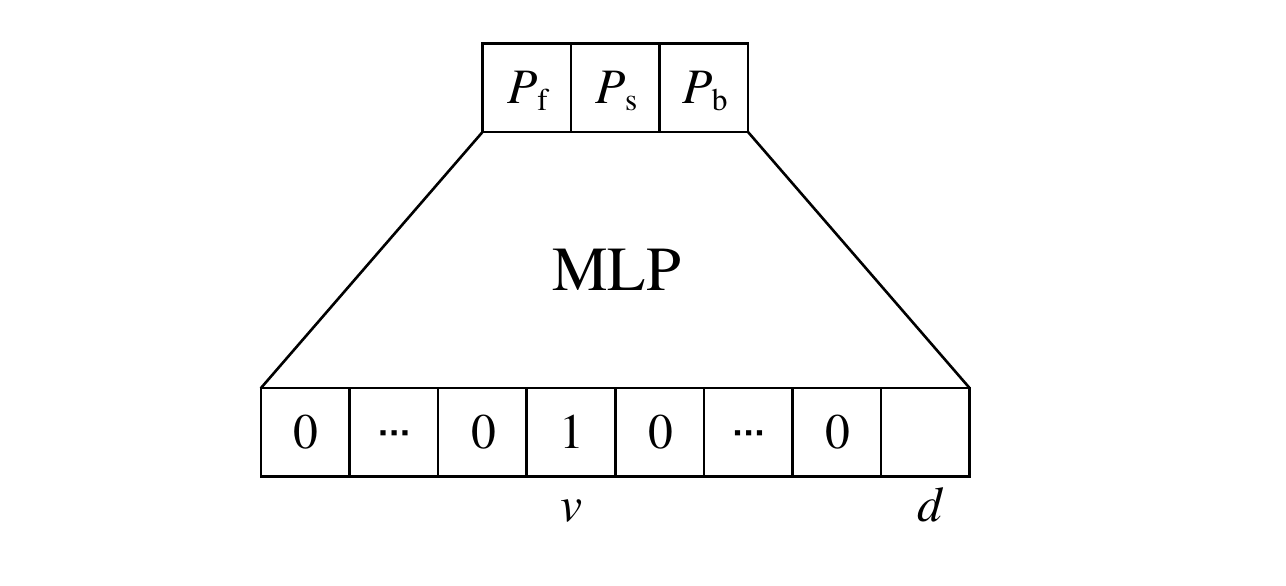}
    \caption{Policy network $\boldsymbol{\pi}$.}
    \label{Fig:Policy}
  \end{figure}

\subsubsection{Policy network}
As we have mentioned, policy function defines a probability distribution over the action space. Based on this idea, we concretize the policy function with the form $\boldsymbol{\pi}(a | s; \boldsymbol{\theta})$, which is the probability of performing action $a \in \mathcal{A}$ at state $s \in \mathcal{S}$, where $\boldsymbol{\theta}$ is the learnable parameters. Particularly, we realize $\boldsymbol{\pi}(a | s; \boldsymbol{\theta})$ as an multilayer perceptron (MLP) network \cite{nielsen2015neural}, as shown in Figure \ref{Fig:Policy}, and then $\boldsymbol{\theta}$ represents the weights in the MLP network that will be learned during the meta-learning. As we can see from Figure \ref{Fig:Policy}, we encode a state $s = (v, d)$ as the input vector to the MLP network, which is the concatenation of a one-hot vector representing node $v$ and a scalar $d$ representing the distance. The output is a 3-dimensional vector generated by Softmax function, where $P_\text{f}$, $P_\text{s}$, and $P_{b}$ are the probabilities of actions $a_\text{f}$, $a_\text{s}$, and $a_{b}$, respectively.

\subsubsection{Transition}
Suppose we sample $K$ sequences (walks) of length $L$ for a node $v$ as its context $C_v$. Let $v$'s $i$-th node sequence $c_v^{(i)} = <v, v_1^{(i)}, \dots, v_L^{(i)}>$ ($1 \le i \le K$), where $v_l^{(i)}$ ($1 \le l \le L$) is the $l$-th sample node of the $i$-th sequence of node $v$. Note that $c_v^{(i)}$ is sampled by the $i$-th transition trajectory of MDP, $\psi_v^{(i)} = <s_0^{(i)} (v)$, $a_0^{(i)}(v)$, $s_1^{(i)}(v)$, $a_1^{(i)}(v)$, $\dots$, $s_{L-1}^{(i)}(v)$, $a_{L-1}^{(i)}(v)$, $s_{L}^{(i)}(v)$, $a_{L}^{(i)}(v)>$, where $s_j^{(i)}(v) \in \mathcal{S}$ and $a_j^{(i)}(v) \in \mathcal{A}$ ($0 \le j \le L$) are the state and action at step $j$, respectively. 
It is easy to see that the probability of state transition from $s_{j-1}^{(i)}(v)$ to $s_j^{(i)}(v)$ under action $a_{j-1}^{(i)}(v) \in \mathcal{A}$ is 1, and hence the probability of $\psi_v^{(i)}$ can be obtained by 
\begin{equation}
\label{Eq:TransitionFirst}
\boldsymbol{\rho} (\psi_v^{(i)};\boldsymbol{\theta}) = \prod_{j=0}^L \boldsymbol{\pi}(a_j^{(i)}(v) | s_j^{(i)}(v); \boldsymbol{\theta}).
\end{equation}
Then the probability of the transition trajectory set $\Psi_v$ of node $v$ is 
\begin{equation}
\label{Eq:TransitionSecond}
\boldsymbol{\rho}(\Psi_v; \boldsymbol{\theta}) = \prod_{i = 1}^{K}\boldsymbol{\rho} (\psi_v^{(i)};\boldsymbol{\theta}). 
\end{equation}
Therefore, we can evaluate the probability of the transition trajectory set $\boldsymbol{\Psi}$ of all nodes by the following equation:
\begin{equation}
\label{Eq:TransitionThird}
\boldsymbol{\rho}(\boldsymbol{\Psi}; \boldsymbol{\theta}) = \prod_{v \in V}\boldsymbol{\rho}(\Psi_v; \boldsymbol{\theta}).
\end{equation}

\renewcommand{\algorithmicrequire}{\textbf{Input:}}
\renewcommand{\algorithmicensure}{\textbf{Output:}}

\begin{algorithm}[t]
\caption{MLANE}
\label{Alg:MLANE}
\begin{algorithmic}[1]
\REQUIRE ~~ \\
Network $G = (V, E)$; Embedding dimensionality $m$; Number of walks per node $K$; Walk length $L$; Sliding window size for SkipGram $w$; Model for analysis task $T$
 \ENSURE ~~ \\
Set of node embeddings $\boldsymbol{Z}$
\STATE Randomly initialize parameters $\boldsymbol{\theta}$ of policy network $\boldsymbol{\pi}$.
\WHILE{policy network $\boldsymbol{\pi}$ does not converge}
\STATE Sample contexts $\boldsymbol{C}$ for all nodes with policy $\boldsymbol{\pi}(\boldsymbol{\theta})$.
\STATE Generate embeddings $\boldsymbol{Z} = \text{SkipGram}(\boldsymbol{C}, w, m)$.
\STATE Train and evaluate $T$ on $\boldsymbol{Z}$, $R = M_T(\boldsymbol{Z})$.
\STATE Update $\boldsymbol{\theta}$ according to Equation (\ref{Eq:PolicyLearning}).
\ENDWHILE
\STATE \textbf{return} $\boldsymbol{Z}$
\end{algorithmic}
\end{algorithm}

\subsubsection{Reward}
As mentioned before, MLANE generates the embeddings $\boldsymbol{Z}$ using SkipGram, i.e., $\boldsymbol{Z}=\text{SkipGram}(\boldsymbol{C}, w, $ $m)$, where $w$ is the sliding window size for SkipGram, and $m$ is the embedding dimension size. The generated embeddings $\boldsymbol{Z}$ will be applied to the given task $T$. We regard the performance $M_T(\boldsymbol{Z})$ of task $T$ over embeddings $\boldsymbol{Z}$ as the final reward $R$, i.e., $R = M_T(\boldsymbol{Z})$.

\subsubsection{Policy Learning}
Different from traditional meta-learning methods which aim at learning the parameters of optimizer, MLANE borrows the idea of meta-learning to learn the parameters $\boldsymbol{\theta}$ of the sampling policy function $\boldsymbol{\pi}$ using policy gradient. The learning objective is defined as:
\begin{equation}
\argmax_{\boldsymbol{\theta}}J(\boldsymbol{\theta}) = \mathbb{E}_{\boldsymbol{\rho}(\boldsymbol{\Psi};\boldsymbol{\theta})} [R].
\end{equation}
It is easy to show that the gradient of $J(\boldsymbol{\theta})$ is 
\begin{equation}
\nabla J(\boldsymbol{\theta}) \propto R\frac{\partial \log(\boldsymbol{\rho}(\boldsymbol{\Psi}; \boldsymbol{\theta})) }{\partial \boldsymbol{\theta}}.
\label{Eq:DerivativeOfJ}
\end{equation}
The the parameters of policy function can be updated as 
\begin{equation}
\boldsymbol{\theta} = \boldsymbol{\theta} + \alpha \nabla J(\boldsymbol{\theta}),
\label{Eq:PolicyLearning}
\end{equation}
where $\alpha$ is learning rate.

\subsection{MLANE and Its Convergence}
Now we can incorporate the sampling strategy learning with the embedding learning into Algorithm \ref{Alg:MLANE}. Essentially MLANE is a meta-learner able to learn parameters $\boldsymbol{\theta}$ of the policy network $\boldsymbol{\pi}$ in an end-to-end fashion. Now we show the convergence of MLANE, for which we first present the following lemma \cite{nesterov1998introductory}: 
\begin{lemma}
\label{Lemma:LowerBound}
If the gradient of function $g$ (i.e. $\nabla g$) is locally Lipschitz continuous with Lipschitz constant $\mathcal{M}$, i.e.,
%, which means that there exists a Lipschitz constant $\mathcal{M}$, such that 
\begin{equation}
\| \nabla g(\boldsymbol{y}) - \nabla g(\boldsymbol{x}) \|_2 \leq \mathcal{M}\| \boldsymbol{y} - \boldsymbol{x} \|_2 ,
\end{equation}
%where $\boldsymbol{x}, \boldsymbol{y} \in \boldsymbol{X}$, $\boldsymbol{X}$ is a local domain of $g$. 
then the following inequation holds: 
\begin{equation}
g(\boldsymbol{y}) \geq  g(\boldsymbol{x}) +  \nabla g(\boldsymbol{x})^T(\boldsymbol{y}-\boldsymbol{x}) - \frac{\mathcal{M}}{2} \| \boldsymbol{y} - \boldsymbol{x} \|_2^2 ,
\label{Eq:LowerBound}
\end{equation}
where $\boldsymbol{x}, \boldsymbol{y} \in \boldsymbol{X}$ and $\boldsymbol{X}$ is a local domain of $g$.
\end{lemma}

Let $\boldsymbol{\theta}^i$ denote the parameters output by the $i$th iteration of MLANE, and according to Equation (\ref{Eq:PolicyLearning}) we have
\begin{equation}
\boldsymbol{\theta}^{i} - \boldsymbol{\theta}^{i-1} = \alpha \nabla J(\boldsymbol{\theta}^{i-1}),
\label{Eq:ParamUpdate}
\end{equation}
where $\alpha$ is the learning rate. Then the following theorem ensures that MLANE will eventually converge to the optimal parameters $\boldsymbol{\theta}^\ast$ after sufficient iterations.

%\begin{theorem}
%\label{Theorem:Converge}
%Assuming that $J$ is locally concave and at point $\boldsymbol{\theta}^\ast$ it reaches the local maximum value, $k$ is the number of iterations. If $\alpha \leq \frac{1}{\mathcal{M}}$, where $\mathcal{M}$ is the Lipschitz constant introduced in Lemma \ref{Lemma:LowerBound}, then:
%\begin{equation}
%J(\boldsymbol{\theta}^\ast) - J(\boldsymbol{\theta}^k)  \leq \frac{\| \boldsymbol{\theta}^{k}-\boldsymbol{\theta}^\ast \|_2^2}{2\alpha k}.
%\label{Eq:Converge}
%\end{equation}
%\end{theorem}

\begin{theorem}
\label{Theorem:Converge}
If the learning rate $\alpha \leq \frac{1}{\mathcal{M}}$, where $\mathcal{M}$ is the Lipschitz constant introduced in Lemma \ref{Lemma:LowerBound}, then:
\begin{equation}
J(\boldsymbol{\theta}^\ast) - J(\boldsymbol{\theta}^k)  \leq \frac{\| \boldsymbol{\theta}^{k}-\boldsymbol{\theta}^\ast \|_2^2}{2\alpha k}.
\label{Eq:Converge}
\end{equation}
\end{theorem}

\begin{proof}
Combining Equations (\ref{Eq:TransitionFirst}), (\ref{Eq:TransitionSecond}), (\ref{Eq:TransitionThird}), and (\ref{Eq:DerivativeOfJ}), we have:
\begin{equation}
\nabla J(\boldsymbol{\theta}) \propto R\sum_{v \in V}\sum_{i = 1}^{K}\sum_{j=0}^L \frac{\partial \log(\boldsymbol{\pi}(a_j^{(i)}(v) | s_j^{(i)}(v); \boldsymbol{\theta})) }{\partial \boldsymbol{\theta}}.
\label{Eq:DerivativeOfJSum}
\end{equation}
Meanwhile, $\boldsymbol{\pi}(a_j^{(i)}(v) | s_j^{(i)}(v); \boldsymbol{\theta})$ can be written as:
\begin{equation}
\boldsymbol{\pi}(a_j^{(i)}(v) | s_j^{(i)}(v); \boldsymbol{\theta}) = \frac{\exp(\boldsymbol{\varphi}(s_j^{(i)}(v), a_j^{(i)}(v))^T\boldsymbol{\theta})}{\sum_{a^\prime \in \mathcal{A}} \exp(\boldsymbol{\varphi}(s_j^{(i)}(v), a^\prime)^T\boldsymbol{\theta})},
\label{Eq:PiSoftmax}
\end{equation}
where $\boldsymbol{\varphi}(s_j^{(i)}(v), a_j^{(i)}(v))$ is the feature vector of state-action pair $(s_j^{(i)}(v)$, $a_j^{(i)}(v))$ \cite{silver2014deterministic}.
By combining Equations (\ref{Eq:DerivativeOfJSum}) and (\ref{Eq:PiSoftmax}), we have:
\begin{equation}
\begin{aligned}
\nabla J(\boldsymbol{\theta}) \propto 
&R\sum_{v \in V}\sum_{i = 1}^{K}\sum_{j=0}^L[\boldsymbol{\varphi}(s_j^{(i)}(v), a_j^{(i)}(v))] \\
&- R\sum_{v \in V}\sum_{i = 1}^{K}\sum_{j=0}^L\sum_{a^\prime \in \mathcal{A}}[\boldsymbol{\varphi}(s_j^{(i)}(v), a^\prime)\boldsymbol{\pi}(a^\prime | s_j^{(i)}(v); \boldsymbol{\theta})]
\label{Eq:LipschitzContinuousOfJ}
\end{aligned}
\end{equation}
From Equation (\ref{Eq:LipschitzContinuousOfJ}), we could notice that the left term of $\nabla J(\boldsymbol{\theta})$ is a constant, while the right term of $\nabla J(\boldsymbol{\theta})$ is a summation over Softmax functions, and hence $\nabla J(\boldsymbol{\theta})$ is Locally Lipschitz continuous \cite{boyd2004convex, gao2017properties}.
%It is easy to show that the derivative of the gradient function $ \nabla J(\boldsymbol{\theta}) $ approaches to a constant as $\boldsymbol{\theta}$ approaches to infinity, which indicates that the derivative of $ \nabla J(\boldsymbol{\theta}) $ is bounded and hence $ \nabla J(\boldsymbol{\theta}) $ is locally Lipschitz continuous \cite{boyd2004convex}. 
Therefore, according to Lemma \ref{Lemma:LowerBound}, Inequation (\ref{Eq:LowerBound}) holds true for $J(\boldsymbol{\theta})$.
%It's easy to see that the gradient of our objective function $J$ (i.e. $\nabla J$) is locally Lipschitz continuous \cite{boyd2004convex} if we fix the architecture and hyperparameters of the model $T$ and SkipGram. 
By substituting Equation (\ref{Eq:ParamUpdate}) into Inequation (\ref{Eq:LowerBound}) with $\boldsymbol{y} = \boldsymbol{\theta}^{i}$, $\boldsymbol{x} = \boldsymbol{\theta}^{i-1}$ and replacing function $g$ with $J$, we get
\begin{equation}
J(\boldsymbol{\theta}^{i}) \geq  J(\boldsymbol{\theta}^{i-1}) +  \alpha(1-\frac{\alpha \mathcal{M}}{2}) \| \nabla J(\boldsymbol{\theta}^{i-1}) \|_2^2.\label{Eq:LowerBoundTheta}
\end{equation}
Since $\alpha \leq \frac{1}{\mathcal{M}}$, the right side of Inequation (\ref{Eq:LowerBoundTheta}) can be shrunk further, which leads to
\begin{equation}
J(\boldsymbol{\theta}^{i}) \geq  J(\boldsymbol{\theta}^{i-1}) + \frac{\alpha}{2} \| \nabla J(\boldsymbol{\theta}^{i-1}) \|_2^2.
\label{Eq:LowerBoundRefined}
\end{equation}
As $J$ is locally concave, then
\begin{equation}
J(\boldsymbol{\theta}^{i-1}) \geq  J(\boldsymbol{\theta}^\ast) + \nabla J(\boldsymbol{\theta}^{i-1})^T(\boldsymbol{\theta}^{i-1}-\boldsymbol{\theta}^\ast).
\label{Eq:ConcaveProp}
\end{equation}
By combining Inequations (\ref{Eq:LowerBoundRefined}) and (\ref{Eq:ConcaveProp}), we have:
\begin{align}
J(&\boldsymbol{\theta}^\ast) - J(\boldsymbol{\theta}^{i}) \nonumber \\ 
& \leq -(\nabla J(\boldsymbol{\theta}^{i-1})^T(\boldsymbol{\theta}^{i-1}-\boldsymbol{\theta}^\ast) + \frac{\alpha}{2} \| \nabla J(\boldsymbol{\theta}^{i-1}) \|_2^2)  \\
& = \frac{1}{2\alpha} (\| \boldsymbol{\theta}^{i-1}-\boldsymbol{\theta}^\ast \|_2^2  -  \| \alpha \nabla J(\boldsymbol{\theta}^{i-1})^T + (\boldsymbol{\theta}^{i-1}-\boldsymbol{\theta}^\ast) \|_2^2). 
\label{Eq:Eight}
\end{align}
Plugging Equation (\ref{Eq:ParamUpdate}) into Inequation (\ref{Eq:Eight}), we obtain:
\begin{equation}
J(\boldsymbol{\theta}^\ast) - J(\boldsymbol{\theta}^{i}) \leq \frac{1}{2\alpha}(\| \boldsymbol{\theta}^{i-1}-\boldsymbol{\theta}^\ast \|_2^2 - \| \boldsymbol{\theta}^i-\boldsymbol{\theta}^\ast \|_2^2).
\end{equation}
By summing over iterations, we finally obtain:
\begin{align}
J(\boldsymbol{\theta}^\ast) - J(\boldsymbol{\theta}^k) & \leq \frac{1}{k}\sum_{i=1}^{k}(J(\boldsymbol{\theta}^\ast) - J(\boldsymbol{\theta}^i)) \\
& \leq \frac{1}{2\alpha k}(\| \boldsymbol{\theta}^{k}-\boldsymbol{\theta}^\ast \|_2^2 - \| \boldsymbol{\theta}^0-\boldsymbol{\theta}^\ast \|_2^2) \\
& \leq \frac{\| \boldsymbol{\theta}^{k}-\boldsymbol{\theta}^\ast \|_2^2}{2\alpha k}.
\label{Eq:Converge}
\end{align}
\end{proof}

Theorem \ref{Theorem:Converge} tells us that if the number of iterations $k$ is big enough, the loss $J(\boldsymbol{\theta}^k)$ can approach to the optimal value $J(\boldsymbol{\theta}^\ast)$ with arbitrary small distance, which confirms the convergence of MLANE.

\begin{table*}[t]
  \renewcommand\arraystretch{1}
  \caption{The statistics of datasets.}
  \begin{center}
	\begin{tabular}{@{}lccccccc@{}}
	\toprule
	\textbf{Datasets} & $|V|$  & $|E|$   & Number of labels & \multicolumn{1}{l}{Node classification}  & \multicolumn{1}{l}{Link prediction} &\multicolumn{1}{l}{Node Clustering} & \multicolumn{1}{l}{Case study}  \\ \midrule
%	Cornell           & 195    & 286     & 5                 &                                         &                                     &                                            &                                   \\
     Citeseer          & 3,327  & 4,676   & 6                 & $\surd$                                 & $\surd$                             & $\surd$                                    &                                 \\
	Cora              & 2,708  & 5,278   & 7                 & $\surd$                                 & $\surd$                             & $\surd$                                    &                                 \\  
%	Polblogs          & 1224   & 16718   & 2                 & $\surd$                                 &                                     &                                \\
%	%PGP              & 10680  & 24316   & none              &                                         & $\surd$                             &                                \\
	BlogCatalog       & 10,312 & 333,983 & 39                & $\surd$                                 &                                     &                                            &                                 \\
	PPI               & 3,890  & 76,584  & 50                & $\surd$                                 &                                     &                                            & $\surd$                         \\
	Amazon            & 3,087  & 2,753   & 23                & $\surd$                                 &                                     &                                            &                                 \\
	HepTh             & 9,877  & 25,998  & none              &                                         & $\surd$                             &                                     &                                 \\ \bottomrule
	\end{tabular}
  \end{center}

  \label{Tbl:Datasets}
\end{table*}

\section{Experiments}
\label{Sec:Experiments}
The objective of experiments is to verify MLANE over tasks of node classification, link prediction, node clustering, and check the adaptiveness of MLANE with case study. The experiments are conducted on a single machine with 128GB RAM and 12 CPU cores at 3.5GHz.

\subsection{Experimental Setting}
\subsubsection{Datasets}
%We evaluate MLANE on seven real world datasets, including three citation networks (Citeseer \cite{sen2008collective}, Cora \cite{sen2008collective}, and HepTh \cite{snapnets}), two social networks (Cornell \cite{wang2017community}, BlogCatalog \cite{Zafarani+Liu:2009}), one e-commercial network (Amazon \cite{he2016ups,mcauley2015image}), and one biology network (PPI \cite{grover2016node2vec}). As these datasets contain small and large, sparse and dense networks, they can reflect the comprehensive characteristics of the network embedding. The statistics and the uses of datasets are summarized in Table \ref{Tbl:Datasets}.
We evaluate MLANE on six real world datasets, including three citation networks (Citeseer$\footnote{\url{http://www.cs.umd.edu/~sen/lbc-proj/LBC.html}\label{umd}}$, Cora$\textsuperscript{\ref{umd}}$, and HepTh$\footnote{\url{http://snap.stanford.edu/data/ca-HepTh.html}}$), one social networks (BlogCatalog$\footnote{\url{http://socialcomputing.asu.edu/datasets/blogcatalog3}}$), one e-commercial network (Amazon$\footnote{\url{https://github.com/librahu/}}$), and one biology network (PPI \cite{grover2016node2vec}). As these datasets contain small and large, sparse and dense networks, they can reflect the comprehensive characteristics of the network embedding. The details of the datasets are described as follows, and their statistics and uses are summarized in Table \ref{Tbl:Datasets}.

\begin {itemize}

\item Citeseer: This is a citation network consisting of 3,327 nodes and 4,676 edges, where a node represents a publication and an edge represents a citation between two publications. Citeseer dataset contains six classes of nodes, including "Agents", "AI", "DB", "IR", "ML", and "HCI".

\item Cora: This is a citation network consisting of 2,708 nodes and 5,278 edges, where a node represents a publication and an edge represents a citation. Cora dataset contains seven classes of nodes, including "Case Based", "Genetic Algorithms", "Neural Networks", "Probabilistic Methods", "Reinforcement Learning", "Rule Learning", and "Theory".

\item BlogCatalog: This is social network consisting of 10,312 nodes and 333,983 edges, where a node represents a blogger in BlogCatalog website, and an edge represents the friendship between two nodes. We use the class labels provided by \cite{Zafarani+Liu:2009} which partitions the nodes into 39 groups.

\item PPI: This is a Protein-Protein Interaction network consisting of 3,890 nodes and 76,584 edges, where a node represents a Homo Sapiens gene, and an edge represents an interaction between two genes. There are fifty classes in PPI, including "Xenobiotic Metabolism", "Fatty Acid Metabolism", "Coagulation", "IR", "EMT”, to name a few.

\item Amazon: This is an e-commerce network consisting 2,753 item nodes and 334 brand nodes, and 2,753 edges, where an edge represents an item belongs to a brand. The item nodes are classified into twenty two types including "Books", "Electronics", "Movies and TV", etc., while the brand nodes constitute one "Brand" type. 

\item HepTh: This is also a citation network consisting of 9,877 nodes and 25,998 edges that are extracted from Arxiv of High Energy Physics-Theory category, where a node also represents a publication and an edge represents a citation between two publications. All the nodes in HepTh are of the same type, and it will be used for link prediction. 

\end{itemize}

%%%%%%%%%%%%%%%%%%%%%%%%%%%%%%%%%%%         Hyperparameters setting table     %%%%%%%%%%%%%%%%%%%%%%%%%%%%%%%%%%%%%%%
\begin{table}[t]
  \renewcommand\arraystretch{}
  \caption{Hyperparameters setting of MLANE.}
  \setlength{\tabcolsep}{3.3mm}
  \centering
  \begin{tabular}{@{}lllll@{}}
    \toprule
    \multicolumn{1}{l|}{Dataset}      & $L$  & $K$  & $m$  & $w$  \\ \midrule
    \multicolumn{5}{c}{Node classification}                       \\ \midrule
    \multicolumn{1}{l|}{Cora}         & 80   & 40   & 128  & 10   \\
    \multicolumn{1}{l|}{Citeseer}     & 80   & 40   & 128  & 10   \\
    \multicolumn{1}{l|}{BlogCatalog}  & 80   & 40   & 128  & 10   \\
    \multicolumn{1}{l|}{PPI}          & 80   & 40   & 128  & 10   \\
    \multicolumn{1}{l|}{Amazon}       & 30   & 10   & 128  & 5    \\ \midrule
    \multicolumn{5}{c}{Link prediction} \\ \midrule
    \multicolumn{1}{l|}{Cora}         & 40   & 10   & 128  & 5    \\
    \multicolumn{1}{l|}{Citeseer}     & 40   & 10   & 128  & 5    \\
    \multicolumn{1}{l|}{HepTh}        & 40   & 10   & 128  & 5    \\ \bottomrule
  \end{tabular}

  \label{Tbl:ParamSetting}
\end{table}
%%%%%%%%%%%%%%%%%%%%%%%%%%%%%%%%%%%%%%%%%%%%%%%%%%%%%%%%%%%%%%%%%%%%%%%%%%%%%%%%%%%%%%%%%%%%%%%%%%%%%%%%%%%%%%%%%%%%%

\subsubsection{Baselines}
We use nine methods as baselines, including seven methods (DeepWalk \cite{perozzi2014deepwalk}, HOPE \cite{ou2016asymmetric}, LINE \cite{tang2015line}, SDNE \cite{wang2016structural},  AttentionWalk \cite{abu2018watch}, ProNE \cite{zhang2019prone}, GAT \cite{velivckovic2017graph}) that preserve homophily, four methods that preserve structural equivalence (struc2vec \cite{ribeiro2017struc2vec}, RiWalk \cite{xuewma2019riwalk}, DRNE \cite{tu2018deep}, Role2Vec\cite{ahmed2018learning}), and one method (node2vec \cite{grover2016node2vec}) that preserves both. Specially, GAT is used only for the node classification task as it needs label information of nodes to supervise the learning of node embeddings. %Besides, since GAT utilizes node features during training while our method doesn't, for fare comparison, we use adjacency matrix of the network instead of node feature matrix.
The baselines are briefly introduced as follows.

\begin{itemize}
 \item DeepWalk \cite{perozzi2014deepwalk}: DeepWalk treats the node sequences sampled by random walks as sentences and feed them into the language model SkipGram to generate the node embeddings. Particularly, DeepWalk uses uniform distribution as the sampling policy for random walks.

 \item node2vec \cite{grover2016node2vec}: Similar to DeepWalk, node2vec also samples nodes sequences through random walks and generates node embeddings by applying SkipGram model. However, node2vec takes two hyperparameters $p$ and $q$ to perform biased random walks for generating nodes sequences.

 \item LINE \cite{tang2015line}: LINE can learn node embeddings that preserve both first-order and second-order proximity among nodes. Particularly, LINE uses BFS strategy, which is more reasonable for preserving second-order proximity.

 \item struc2vec \cite{ribeiro2017struc2vec}: This method first constructs a multi-layer graph to encode the structural similarity between nodes, then generates the node contexts by random walks over the multi-layer graph, and finally feeds the generated contexts into SkipGram for learning node embeddings. 

 \item HOPE \cite{ou2016asymmetric}: HOPE preserves the asymmetric transitivity between nodes by approximating the high-order proximity. It first constructs a matrix reflecting asymmetric transitivity, and then learns node embeddings by SVD on the constructed matrix.

 \item SDNE \cite{wang2016structural}: SDNE learns node embeddings through a deep auto-endcoder with a supervised component to preserve the first-order proximity between nodes and an unsupervised component to preserve the second-order proximity between nodes.
 
 \item GAT \cite{velivckovic2017graph}: Graph Attention Networks (GAT) are a variant of Graph Convolutional Network (GCN) \cite{kipf2016semi}, which can learn a representation for a node with attention to its neighborhood.  
 
 %adopts attention mechanisms to learn the weights between two connected nodes so that nodes could learn more information from their important neighbors.

 \item AttentionWalk \cite{abu2018watch}: AttentionWalk is a graph attention model that could learn the $walk\  length$ of random walk automatically through the graph attention mechanism.

 \item ProNE \cite{zhang2019prone}:  ProNE is a fast and scalable network embedding method, and it learns the node embeddings through spectral propagation.

 \item RiWalk \cite{xuewma2019riwalk}: RiWalk first adopts a graph kernel method to identify the role of nodes, and then applys random walks to sample nodes having similar roles for generating node sequences before using SkipGram model to learn node embeddings.

 \item Role2Vec \cite{ahmed2018learning}: Role2Vec first performs attributed random walks with the help of motifs to sample node sequences, and then also generates node embedding by SkipGram model.

 \item DRNE \cite{tu2018deep}: DRNE samples 1-hop neighbors and adopts LSTM to preserve the regular equivalence between nodes in a recursive way that aggregates embeddings of the neighboring nodes into the embedding of the target node.
\end{itemize}

%%%%%%%%%%%%%%%%%%%%%%%%%%%%%%%%%%%        Node classification results table     %%%%%%%%%%%%%%%%%%%%%%%%%%%%%%%%%%%%%%
\begin{table*}[t]
  \renewcommand\arraystretch{}
  \caption{Performance of node classification.}
  \begin{center}
	\begin{tabular}{@{}lcccccccccc@{}}
	\toprule
	\multicolumn{1}{l|}{ }    & \multicolumn{5}{c|}{$Micro\text{-}F1$}                                                                                         & \multicolumn{5}{c}{$Macro\text{-}F1$}                                                                     \\ \midrule
	\multicolumn{1}{l|}{Dataset}    & Cora           & Citeseer       & BlogCatalog    & PPI           & \multicolumn{1}{c|}{Amazon}         & Cora           & Citeseer       & Blogcatalog    & PPI            & Amazon         \\ \midrule
	\multicolumn{11}{c}{Homophily Preserving Methods}                                                                                                                                                                                                               \\ \midrule
	\multicolumn{1}{l|}{DeepWalk}  & 0.811          & 0.578          & 0.406          & 0.218          & \multicolumn{1}{c|}{0.828}          & 0.800          & 0.531          & 0.262          & 0.186          & 0.113          \\
	\multicolumn{1}{l|}{HOPE}      & 0.491          & 0.520          & 0.188          & 0.134          & \multicolumn{1}{c|}{0.880}          & 0.376          & 0.456          & 0.041          & 0.066          & 0.118          \\
	\multicolumn{1}{l|}{LINE}      & 0.696          & 0.468          & 0.369          & 0.204          & \multicolumn{1}{c|}{0.891}          & 0.694          & 0.428          & 0.237          & 0.172          & 0.120          \\
     \multicolumn{1}{l|}{AttentionWalk}        & 0.681          & 0.564          & 0.186          & 0.104          & \multicolumn{1}{c|}{0.790}          & 0.646          & 0.520          & 0.037          & 0.052          & 0.108          \\
     \multicolumn{1}{l|}{ProNE}     & 0.812          & 0.593          & 0.412          & 0.229          & \multicolumn{1}{c|}{0.844}          & 0.808          & 0.531          & 0.253          & 0.191          & 0.122          \\
	\multicolumn{1}{l|}{SDNE}      & 0.701          & 0.484          & 0.401          & 0.197          & \multicolumn{1}{c|}{0.886}          & 0.689          & 0.431          & 0.257          & 0.176          & 0.119          \\ 
     \multicolumn{1}{l|}{GAT}       & 0.816          & 0.535          & 0.407          & 0.224          & \multicolumn{1}{c|}{0.883}          & 0.806          & 0.487          & 0.264          & 0.183          & 0.121          \\ \midrule
	\multicolumn{11}{c}{Structural Equivalence Preserving Methods}                                                                                                                                                                                                  \\ \midrule
	\multicolumn{1}{l|}{struc2vec} & 0.297          & 0.281          & 0.132          & 0.086          & \multicolumn{1}{c|}{0.880}          & 0.161          & 0.244          & 0.044          & 0.070          & 0.127          \\
	\multicolumn{1}{l|}{RiWalk}    & 0.498          & 0.344          & 0.176          & 0.099          & \multicolumn{1}{c|}{0.885}          & 0.440          & 0.294          & 0.040          & 0.064          & 0.119          \\
	\multicolumn{1}{l|}{DRNE}      & 0.282          & 0.234          & 0.172          & 0.081          & \multicolumn{1}{c|}{0.850}          & 0.063          & 0.107          & 0.035          & 0.033          & 0.115          \\
	\multicolumn{1}{l|}{Role2Vec}  & 0.795          & 0.544          & 0.281          & 0.152          & \multicolumn{1}{c|}{0.861}          & 0.794          & 0.505          & 0.148          & 0.124          & 0.127          \\ \midrule
	\multicolumn{11}{c}{Both Properties Preserving Methods}                                                                                                                                                                                                         \\ \midrule
	\multicolumn{1}{l|}{node2vec}  & 0.816          & 0.594          & 0.411          & 0.229          & \multicolumn{1}{c|}{0.889}          & 0.808          & 0.543          & 0.275          & 0.188          & 0.125          \\
	\multicolumn{1}{l|}{MLANE}     & \textbf{0.832} & \textbf{0.624} & \textbf{0.416} & \textbf{0.232} & \multicolumn{1}{c|}{\textbf{0.897}} & \textbf{0.832} & \textbf{0.573} & \textbf{0.282} & \textbf{0.197} & \textbf{0.172} \\ \bottomrule
	\end{tabular}
  \end{center}
  \label{Tbl:Classification}
\end{table*}
%%%%%%%%%%%%%%%%%%%%%%%%%%%%%%%%%%%%%%%%%%%%%%%%%%%%%%%%%%%%%%%%%%%%%%%%%%%%%%%%%%%%%%%%%%%%%%%%%%%%%%%%%%%%%%%%%%%%%

\subsubsection{Metrics}
Similar to existing works \cite{wang2016structural}, we use $Micro\text{-}F1$ and $Macro\text{-}F1$ as the metrics for node classification, and $precision\text{@}k$ 
%and $MAP@k$ (Mean Average Precision)
for link prediction, and $purity$ and the $NMI$ (Normalized Mutual Information) for node clustering. To save the space, here we just give the definitions of $purity$ and $NMI$, and one can refer to \cite{wang2016structural} for the definitions of the remaining metrics. %Note that we only use one metric as the training reward and train our model once for each task as it is enough for the model to achieve the best performance in the task.

\begin{itemize}

  \item $purity$ is defined as:
    \begin{equation*}
      \begin{split}
        purity = \frac{1}{n}\sum_{i=1}^{t}\max_j|cluster(i)\cap class(j)|
      \end{split}
    \end{equation*}
    where $n$ is the total number of clustered points, $t$ is the number of clusters, $cluster(i)$ is the set of points in cluster $i$, and $class(j)$ is the set of points belonging to ground truth class $j$. 

  \item $NMI$ is defined as:
    \begin{equation*}
    NMI = \frac{MI(C, C^\prime)}{\sqrt{H(C)H(C^\prime)}},
    \end{equation*}
    where 
    \begin{equation*}
      \begin{split}
        &MI(C,C^\prime) = \sum_{i=1}^{|C|}\sum_{j=1}^{|C^\prime|}P(i,j)\log(\frac{P(i,j)}{P(i)P^\prime(j)}), \\
        &H(C) = -\sum_{i=1}^{|C|}P(i)\log(P(i)), \\
        &H(C^\prime) = -\sum_{j=1}^{|C^\prime|}P^\prime(j)\log(P^\prime(j)),
      \end{split}
    \end{equation*}
    $C$ is the set of clusters, $C^\prime$ is the set of ground truth classes, $P(i) = |cluster(i)|/n$, $P^\prime(j) = |class(j)|/n$, and $P(i,j) = |cluster(i)\cap class(j)|/n$.
\end{itemize}

\subsubsection{Hyperparameter Setting}
\label{Sec:ParamSetting}
%In our experiments, we set walk length $L = 80$, number of walks per node $K = 40$, window size $w = 10$. 
In our experiments, we set embedding dimensionality $m = 128$ for our method and all the baselines. For the hyperparameters of all the baselines, we manully turn them to be optimal. We implement our policy network $\boldsymbol{\pi}$ with 2 hidden layers consisting 10 and 5 perceptrons respectively. The learning rate $\alpha$ of our method is set to $0.002$. Specially, the hyperparameters setting of MLANE are listed in Table \ref{Tbl:ParamSetting}. For fairness, the hyperparameters of baseline methods are tuned for their best performance on the datasets, for details of which one can access https://github.com/7733com/MLANE.

%%%%%%%%%%%%%%%%%%%%%%%%%%%%%%%%%%%%%%%%%%%%%%%%%   Link prediction tables and figures  %%%%%%%%%%%%%%%%%%%%%%%%%%%%%%%%%%%%%%%%%%%%%%%%%
\begin{figure*}[t]
  \centering
  \subfigure[\textbf{Citeseer}]{
    \begin{minipage}[t]{0.25\linewidth}
      \centering
      \includegraphics[width=\linewidth]{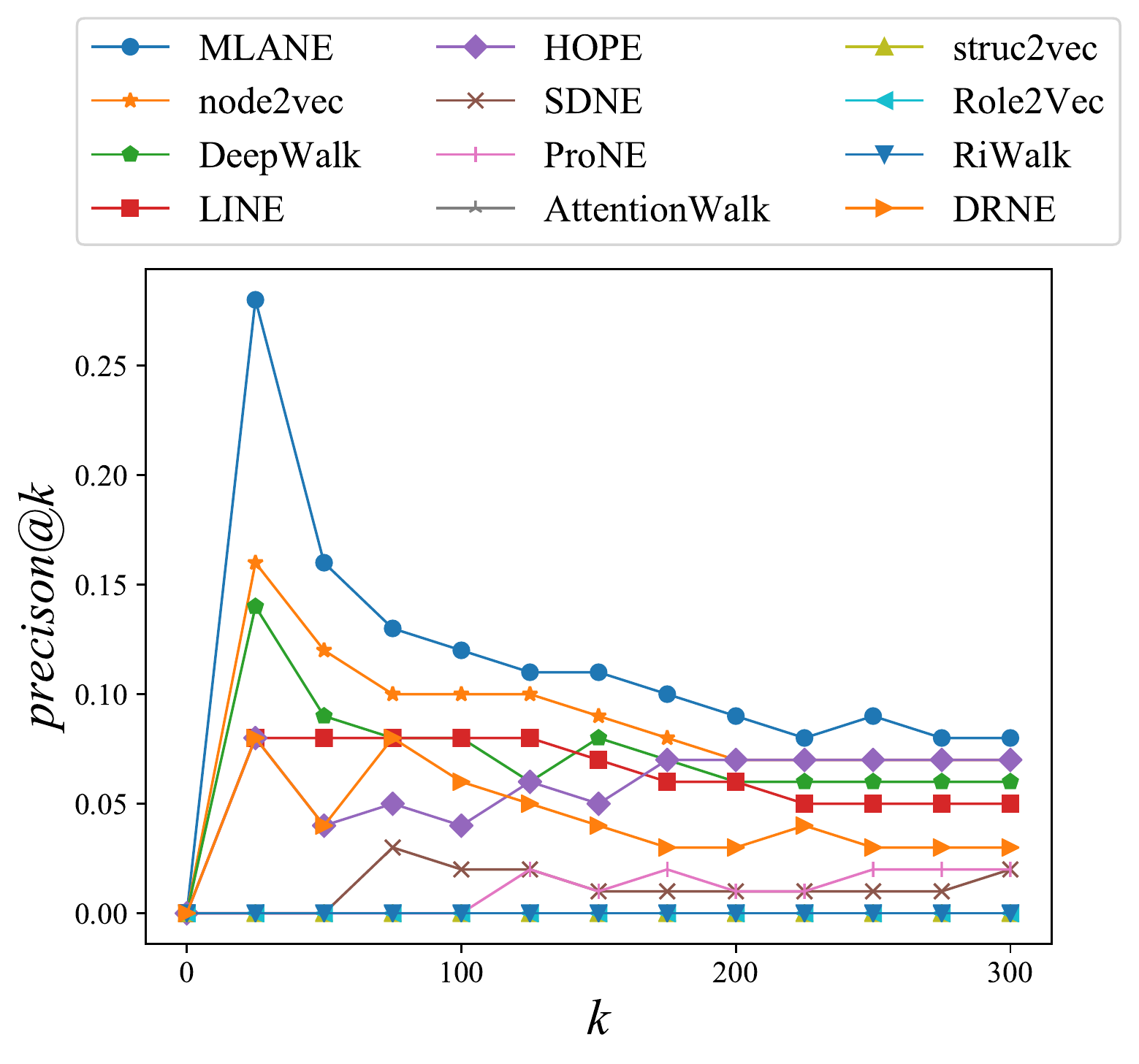}
      %\caption{fig1}
    \end{minipage}%
  }%
  \subfigure[\textbf{Cora}]{
    \begin{minipage}[t]{0.25\linewidth}
      \centering
      \includegraphics[width=\linewidth]{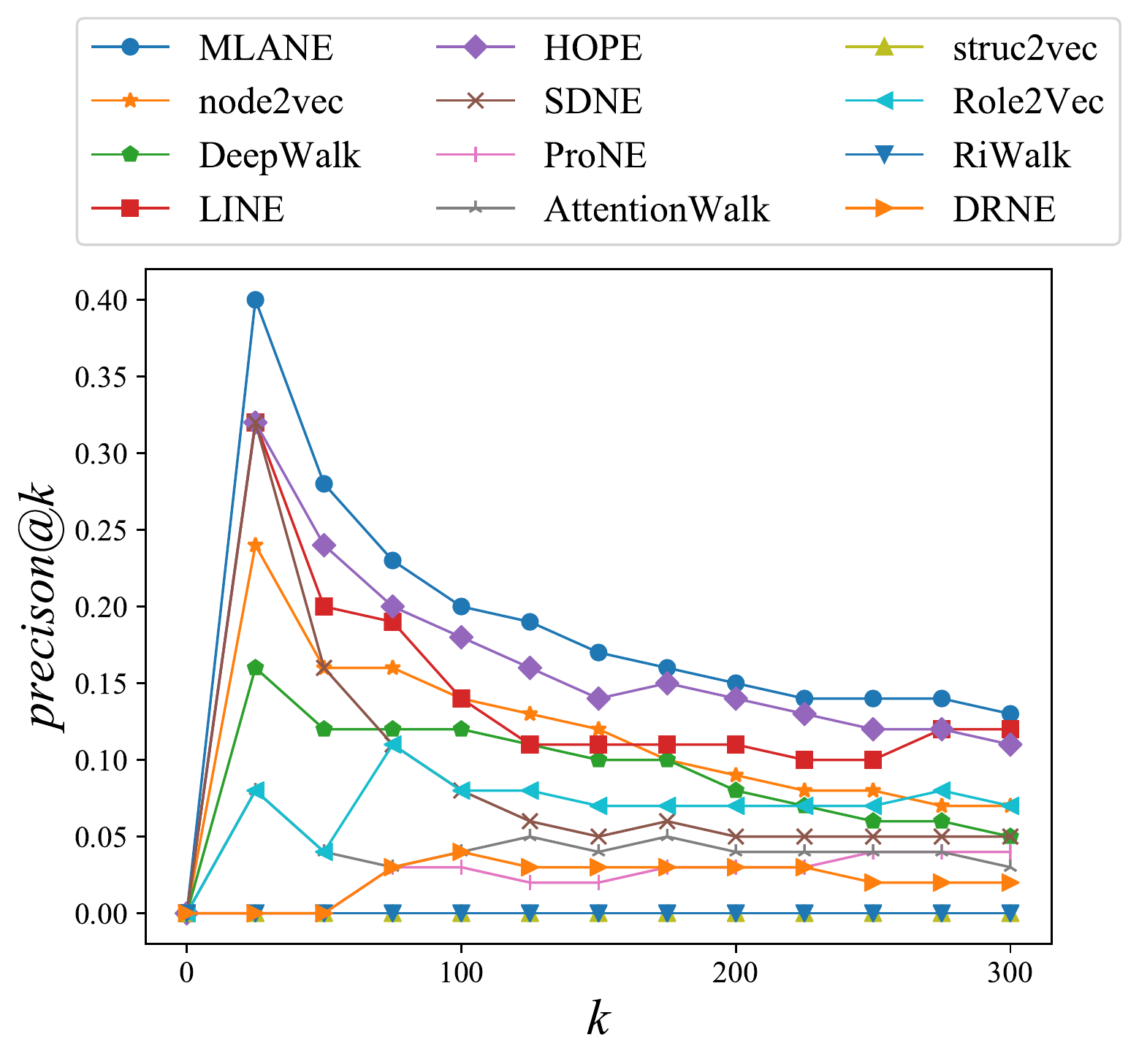}
      %\caption{fig2}
    \end{minipage}%
  }%
  \quad
  \subfigure[\textbf{HepTh}]{
    \begin{minipage}[t]{0.25\linewidth}
      \centering
      \includegraphics[width=\linewidth]{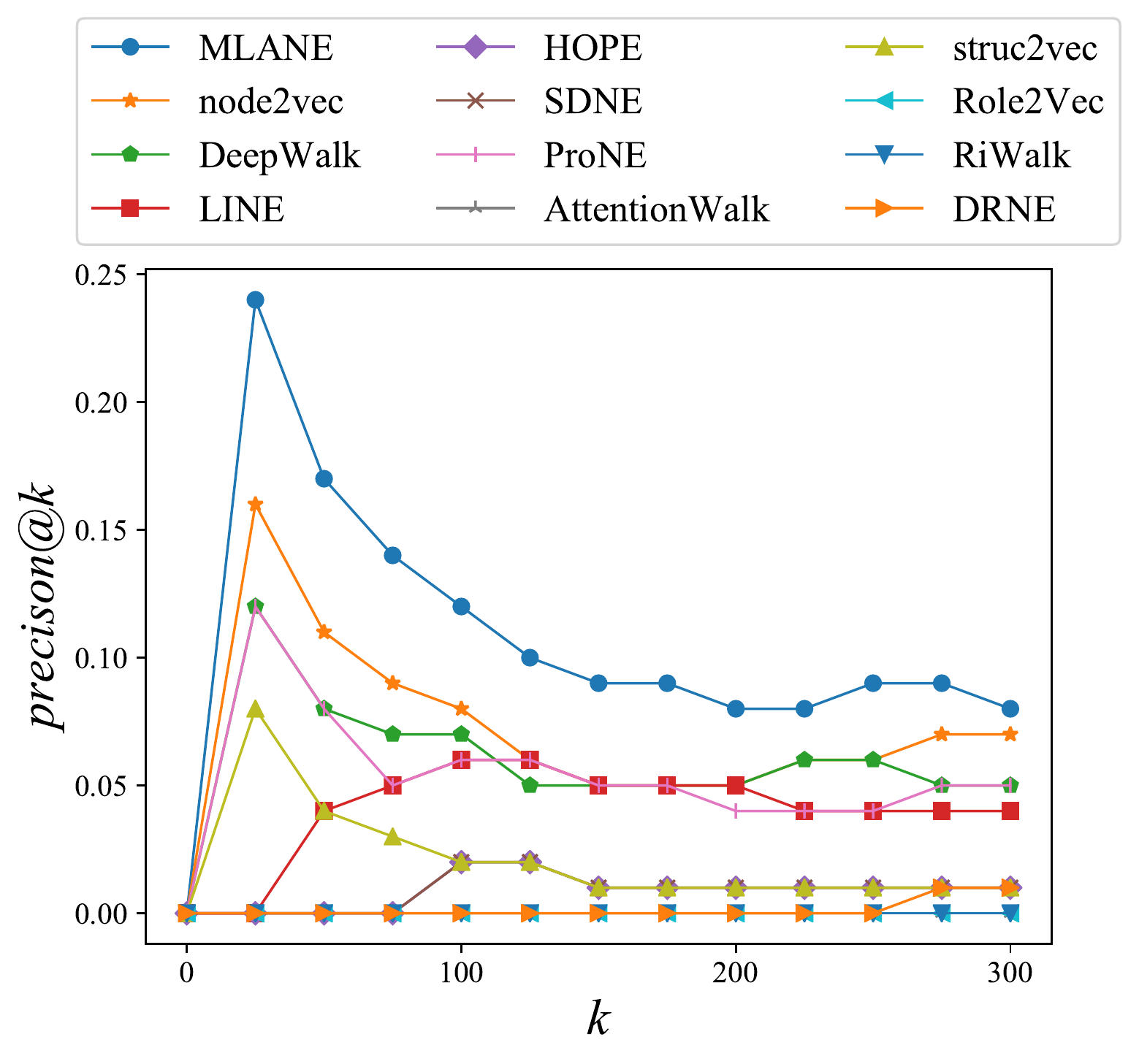}
      %\caption{fig2}
    \end{minipage}
  }%
  \centering
  \caption{Link prediction results in terms of $precision@k$.}
  \label{Fig:precision@k}
\end{figure*}

%%%%%%%%%%%%%%%%%%%%%%%%%%%%%%%%%%%%%%%%%%%%%%%%%   Node clustering results %%%%%%%%%%%%%%%%%%%%%%%%%%%%%%%%%%%%%%%%%%%%%%%%%%%%%%%%%%%%%%%
\begin{table}[t]
  \caption{Node clustering results in terms of $purity$ and $NMI$.}
  \begin{center}
  \begin{tabular}{@{}l|cc|cc@{}}
  \toprule
             & \multicolumn{2}{c|}{$purity$} & \multicolumn{2}{c}{$NMI$} \\ \midrule
    Method   & Cora               & Citeseer          & Cora            & Citeseer        \\ \midrule  
    MLANE    & \textbf{0.913}     & \textbf{0.600}    & \textbf{0.677}  & \textbf{0.321}  \\
    LINE     & 0.384              & 0.325             & 0.093           & 0.048           \\
    node2vec & 0.760              & 0.495             & 0.425           & 0.278           \\ \bottomrule
  \end{tabular}
  \end{center}

  \label{Tbl:NodeClustering}
\end{table}
\begin{figure*}[t]
  \centering
  \subfigure[Sampling for the purple (IR) node]{
    \label{Fig:SampleVisualizationIR}
    \begin{minipage}[t]{0.25\linewidth}
      \centering
      \includegraphics[width=\linewidth]{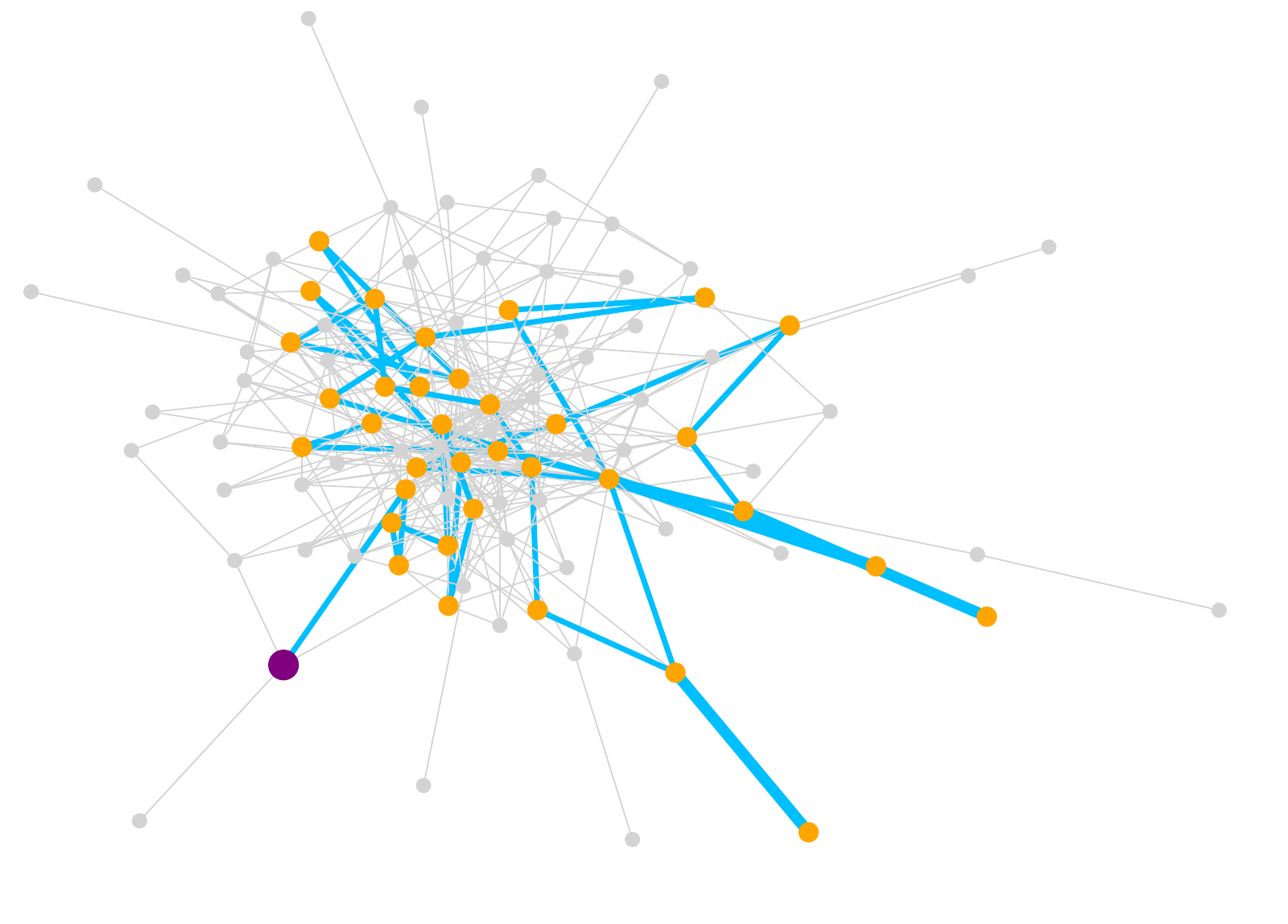}
    \end{minipage}%
  }
  \subfigure[Sampling for the red (EMT) node]{
    \label{Fig:SampleVisualizationEMT}
    \begin{minipage}[t]{0.25\linewidth}
      \centering
      \includegraphics[width=\linewidth]{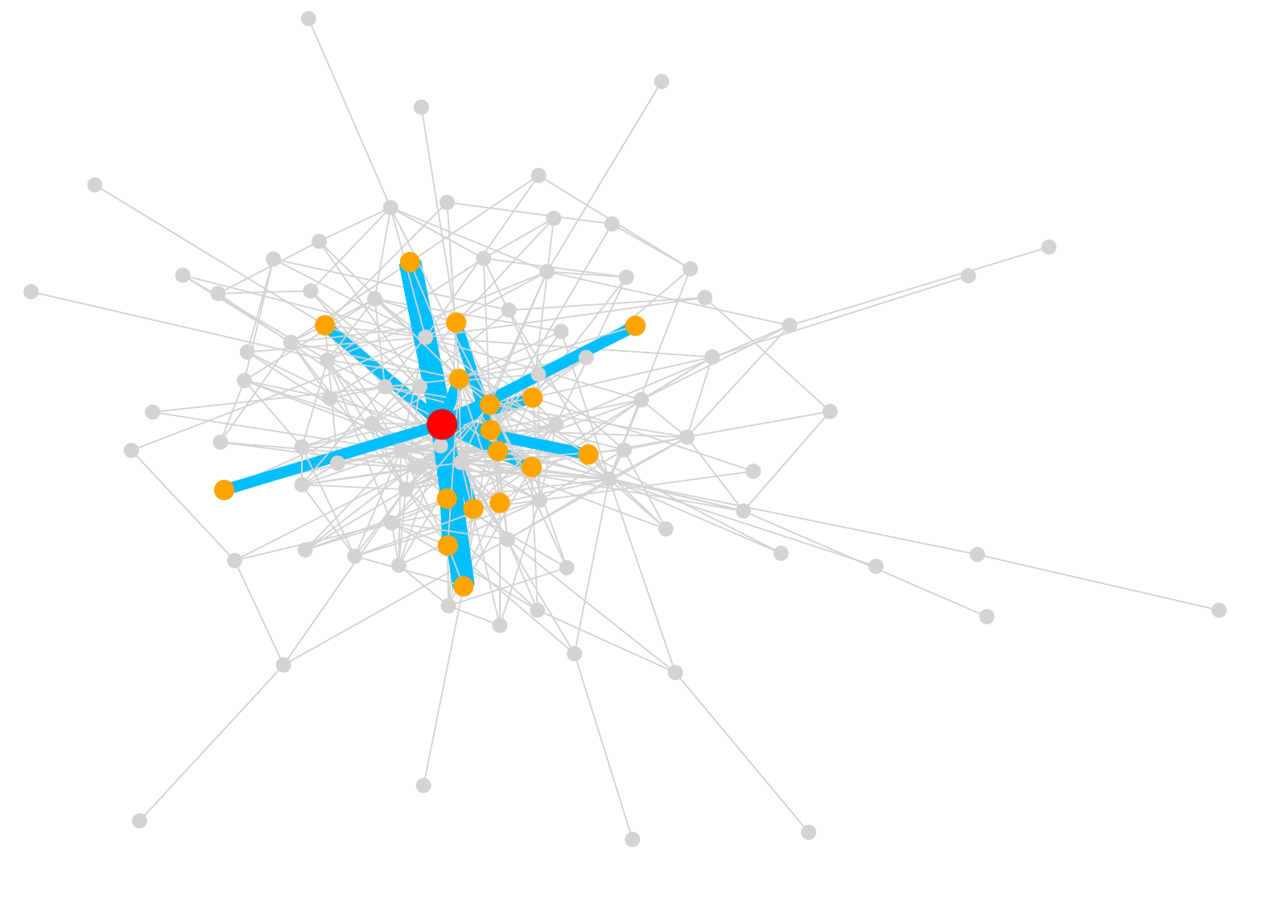}
    \end{minipage}
  }
  \subfigure[Sampling for the green (Coagulation) node]{
    \label{Fig:SampleVisualizationCoa}
    \begin{minipage}[t]{0.25\linewidth}
      \centering
      \includegraphics[width=\linewidth]{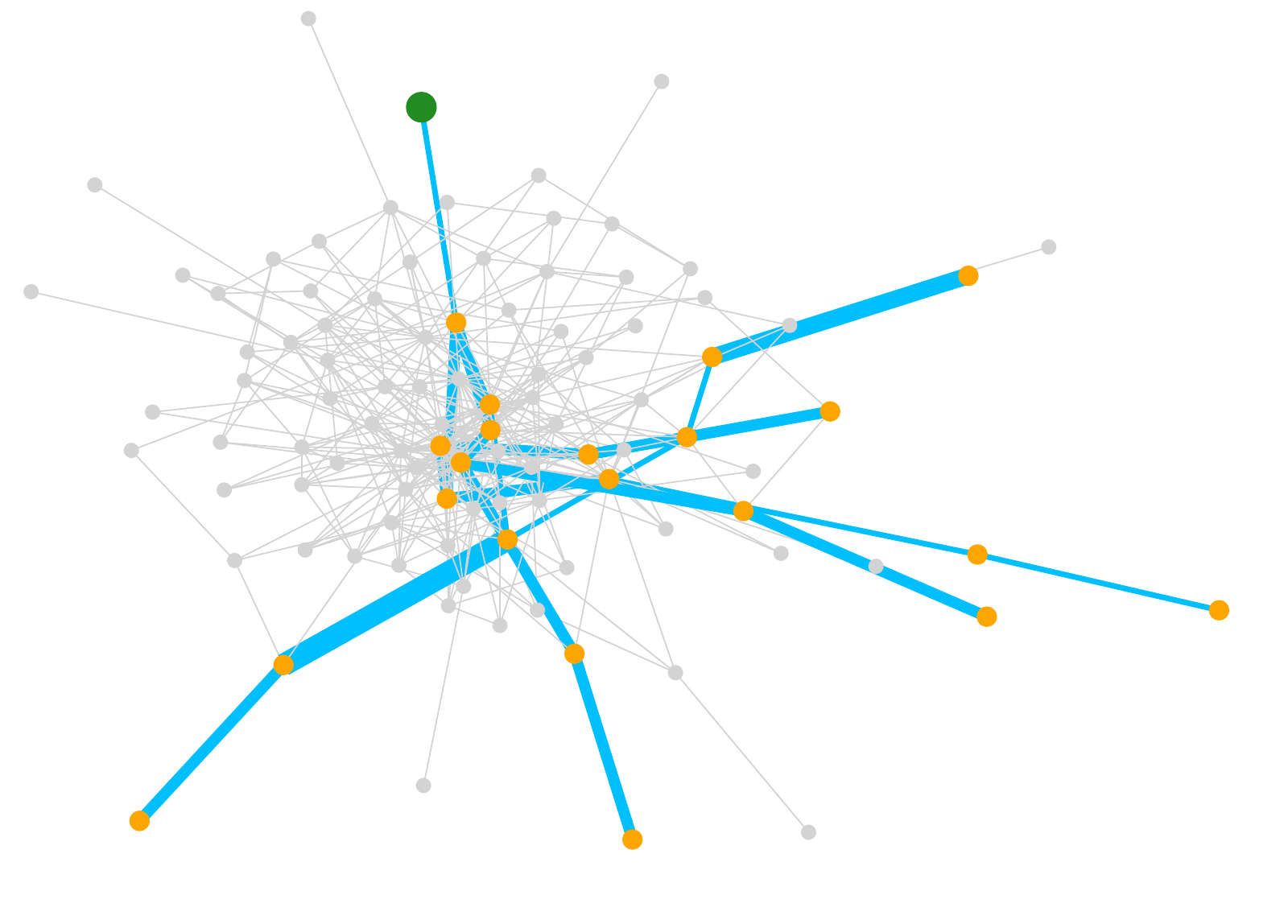}
    \end{minipage}
  }
  \centering
  \caption{Visualization of sampling process for different types of nodes.}
  \label{Fig:SampleVisualization}
\end{figure*}

%%%%%%%%%%%%%%%%%%%%%%%%%%%%%%%%%%%%%%%%%%%%%%%%%%%%%%%%%%%%%%%%%%%%%%%%%%%%%%%%%%%%%%%%%%%%%%%%%%%%%%%%%%%%%%%%%%%%%%%%%%%%%%%%%%%%%%%%%

\subsection{Node Classification}
For node classification, we use the node embeddings generated by MLANE and baselines as feature vectors of nodes, and feed them into a one-vs-rest logistic regression classifier which is trained by using LIBLINEAR package \cite{fan2008liblinear}, and we train our model using $Macro\text{-}F1$ as the reward. On each dataset used for node classification, we randomly sample 80\% of nodes as the training set and use the left nodes as testing set. Note that we don't use HepTh here as it contains no label information. From Table \ref{Tbl:Classification}, we can see that MLANE outperforms the baselines on all datasets. Especially, on Amazon dataset MLANE achieves 35.4\% improvement over the second best method in terms of $Macro\text{-}F1$. In Amazon dataset, there are 22 classes of item nodes which is defined mostly by homophily, and one class of brand nodes which is defined mostly by structural equivalence. Therefore, the results demonstrate that MLANE is able to adaptively capture homophily and structural equivalence with different weights for the embeddings of different nodes.

\subsection{Link Prediction}
We evaluate the performance of MLANE on link prediction task in terms of $precision\text{@}k$.
% and $MAP\text{@}k$
On each dataset, we randomly remove 10\% links and use the remaining network to learn the node embeddings. Similar to existing works \cite{wang2016structural}, a link between two nodes is predicted if the similarity (evaluated by inner product) of the embeddings of that two nodes ranks in top $k$. For link prediction, we train our model using $precision@k$ as the reward. Due to the space limitation, we only show the results on Citeseer, Cora, and HepTh. 
%The $precision\text{@}k$ and $MAP\text{@}k$ are shown in Figure \ref{Fig:precision@k} and Table \ref{Tbl:MAP@k}, respectively, 
The $precision\text{@}k$ are shown in Figure \ref{Fig:precision@k}, from which we can see that MLANE also outperforms the baselines on link prediction. Interestingly, in our experiments almost all baselines that only preserve structural equivalence exhibit relatively poor performance on the link prediction task, which is reasonable as homophily is usually more important than structural equivalence for link prediction. At the same time, we see that MLANE shows better performance than baselines that only preserve homophily. This is because some links can not be explained by homophily, which MLANE can adaptively realize during the sampling process with learned sampling strategy.

\subsection{Node Clustering}
We conduct node clustering by applying $K$-means algorithm, of which the results are visualized via the dimensionality reduction algorithm t-SNE \cite{maaten2008visualizing}. For node clustering, we train our model using $NMI$ as the reward. Because of the space limitation, we only show the results on Citeseer and Cora, and for clear visualization, we just show three clusters corresponding to three classes for each dataset. Particularly, Figure \ref{Fig:NodeClusteringCiteseer} shows the clustering results on Citeseer, where the nodes of classes ``IR'', ``AI'', and ``Agents'' are marked with colors blue, green, and red, respectively. Similarly, Figure \ref{Fig:NodeClusteringCora} shows the clustering results on Cora, where nodes of classes ``Case Based'', ``Genetic Algorithm'', ``Reinforcement Learning'' are also marked with colors blue, green, and red, respectively. Figure \ref{Fig:NC_Citeseer} and Figure \ref{Fig:NC_Cora} are the original distributions of nodes of Citeseer and Cora, respectively. Intuitively, we can see that on each dataset the clusters produced by MLANE (Figures \ref{Fig:NC_Citeseer_MLANE} and \ref{Fig:NC_Cora_MLANE}) are purer than those produced by LINE and node2vec (Figures \ref{Fig:NC_Citeseer_LINE},  \ref{Fig:NC_Citeseer_node2vec}, \ref{Fig:NC_Cora_LINE}, and \ref{Fig:NC_Cora_node2vec}). To verify our intuition, we further quantitatively evaluate the clustering in terms of the metrics $purity$ and the $NMI$ (Normalized Mutual Information), which are shown in Table \ref{Tbl:NodeClustering}. From Table \ref{Tbl:NodeClustering} we can see that MLANE produces a remarkably better performance than LINE and node2vec.

%%%%%%%%%%%%%%%%%%%%%%%%%%%%%%%%%%%%%%%%%%%%%%%%%   Node clustering figures  %%%%%%%%%%%%%%%%%%%%%%%%%%%%%%%%%%%%%%%%%%%%%%%%%
\begin{figure*}[t]
  \centering
  \subfigure[\textbf{Original distribution}]{
    \begin{minipage}[t]{0.25\linewidth}
      \centering
      \label{Fig:NC_Citeseer}
      \includegraphics[width=\linewidth]{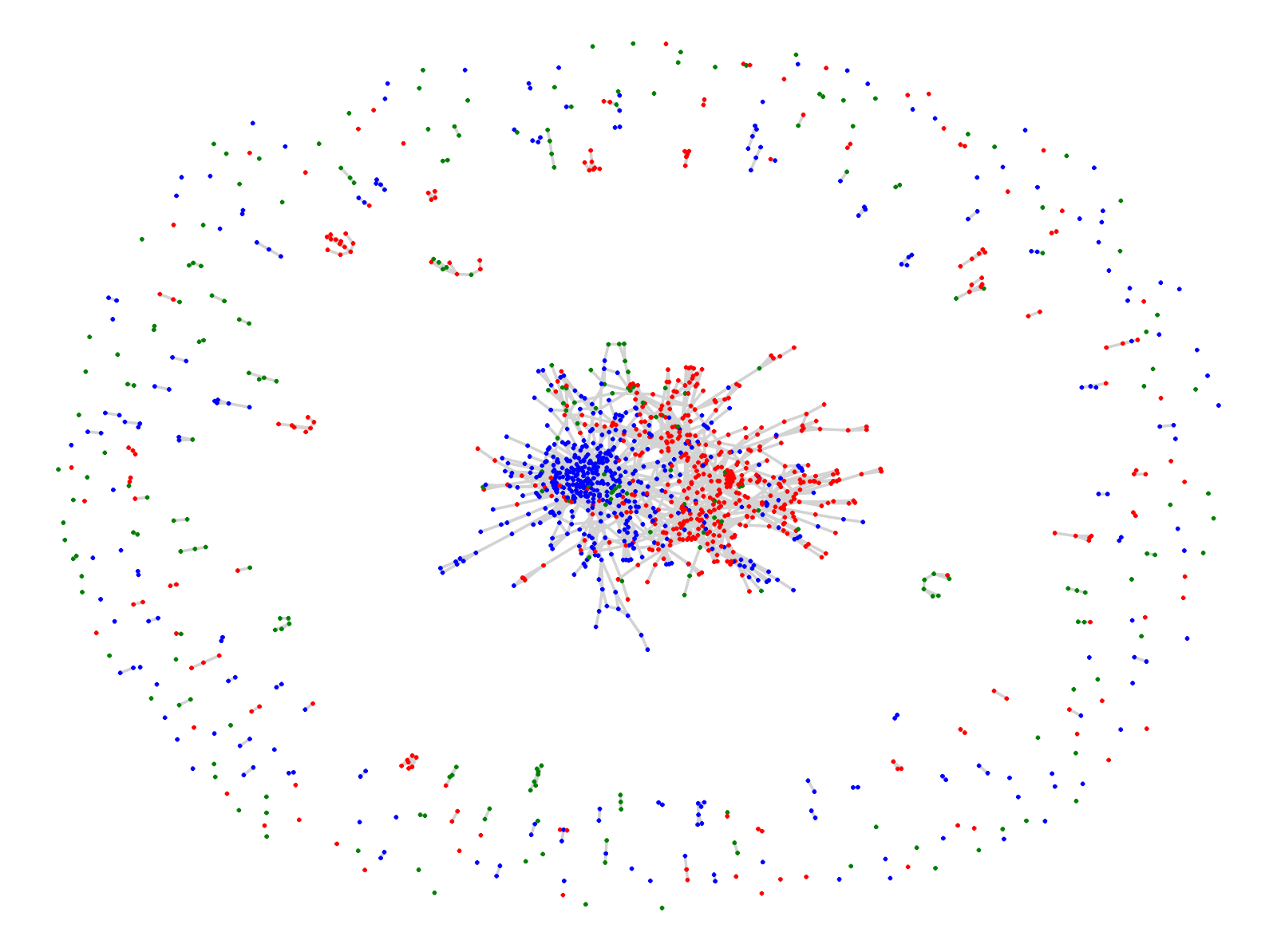}
      %\caption{fig1}
    \end{minipage}%
  }%
  \subfigure[\textbf{MLANE}]{
    \begin{minipage}[t]{0.25\linewidth}
      \centering
      \label{Fig:NC_Citeseer_MLANE}
      \includegraphics[width=\linewidth]{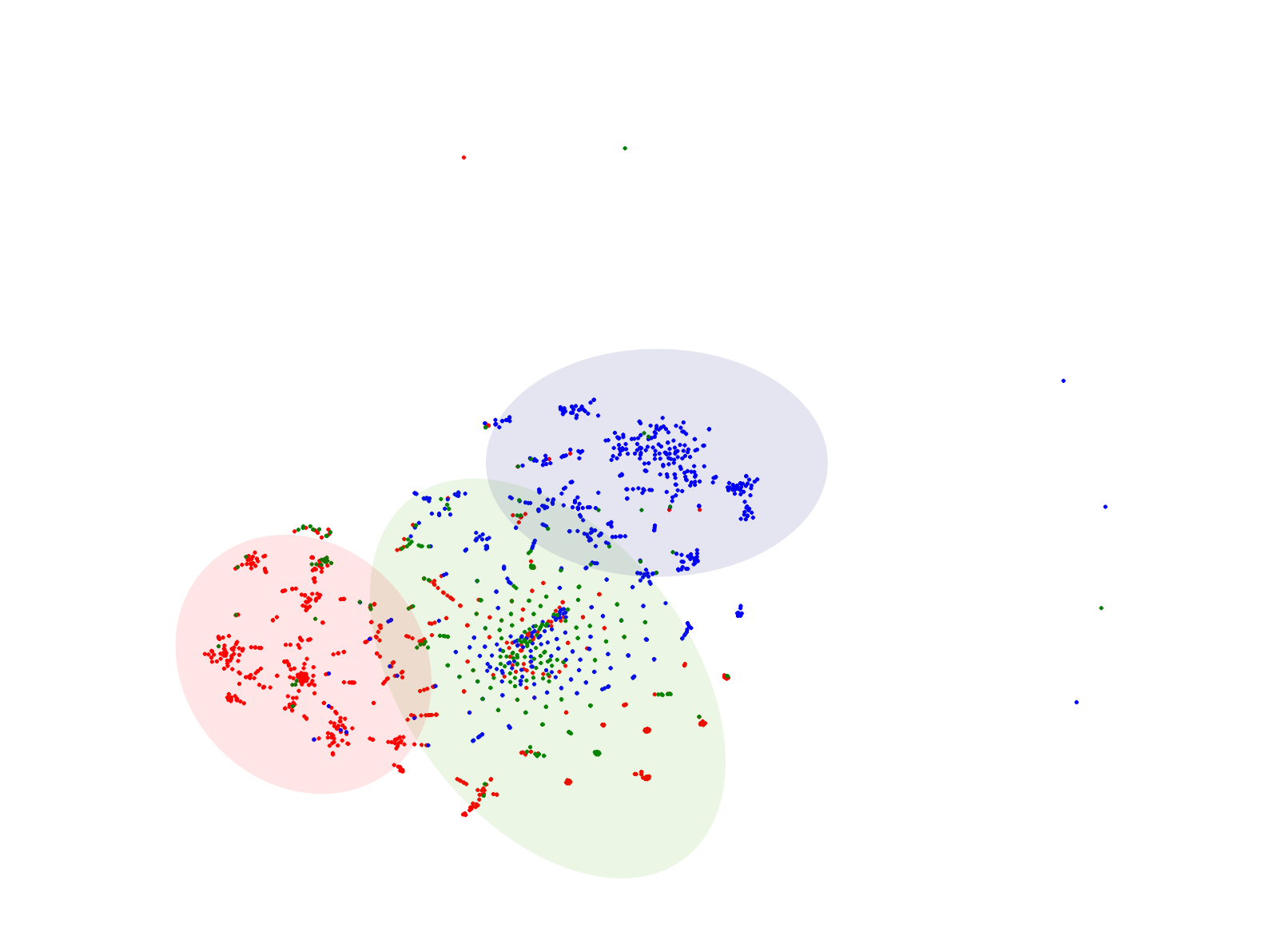}
      %\caption{fig2}
    \end{minipage}%
  }%
  \subfigure[\textbf{LINE}]{
    \begin{minipage}[t]{0.25\linewidth}
      \centering
      \label{Fig:NC_Citeseer_LINE}
      \includegraphics[width=\linewidth]{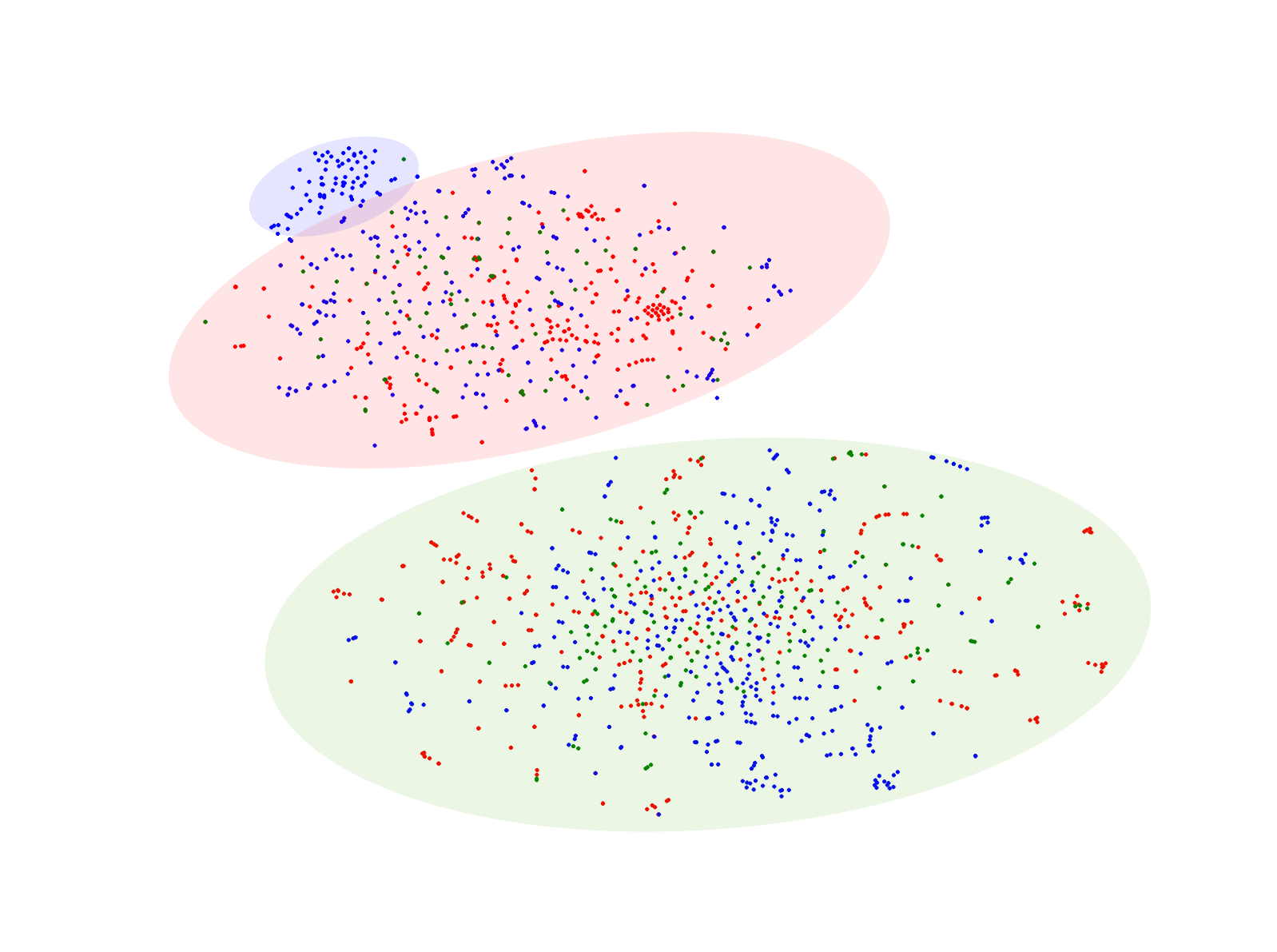}
      %\caption{fig2}
    \end{minipage}
  }%
  \subfigure[\textbf{node2vec}]{
    \begin{minipage}[t]{0.25\linewidth}
      \centering
      \label{Fig:NC_Citeseer_node2vec}
      \includegraphics[width=\linewidth]{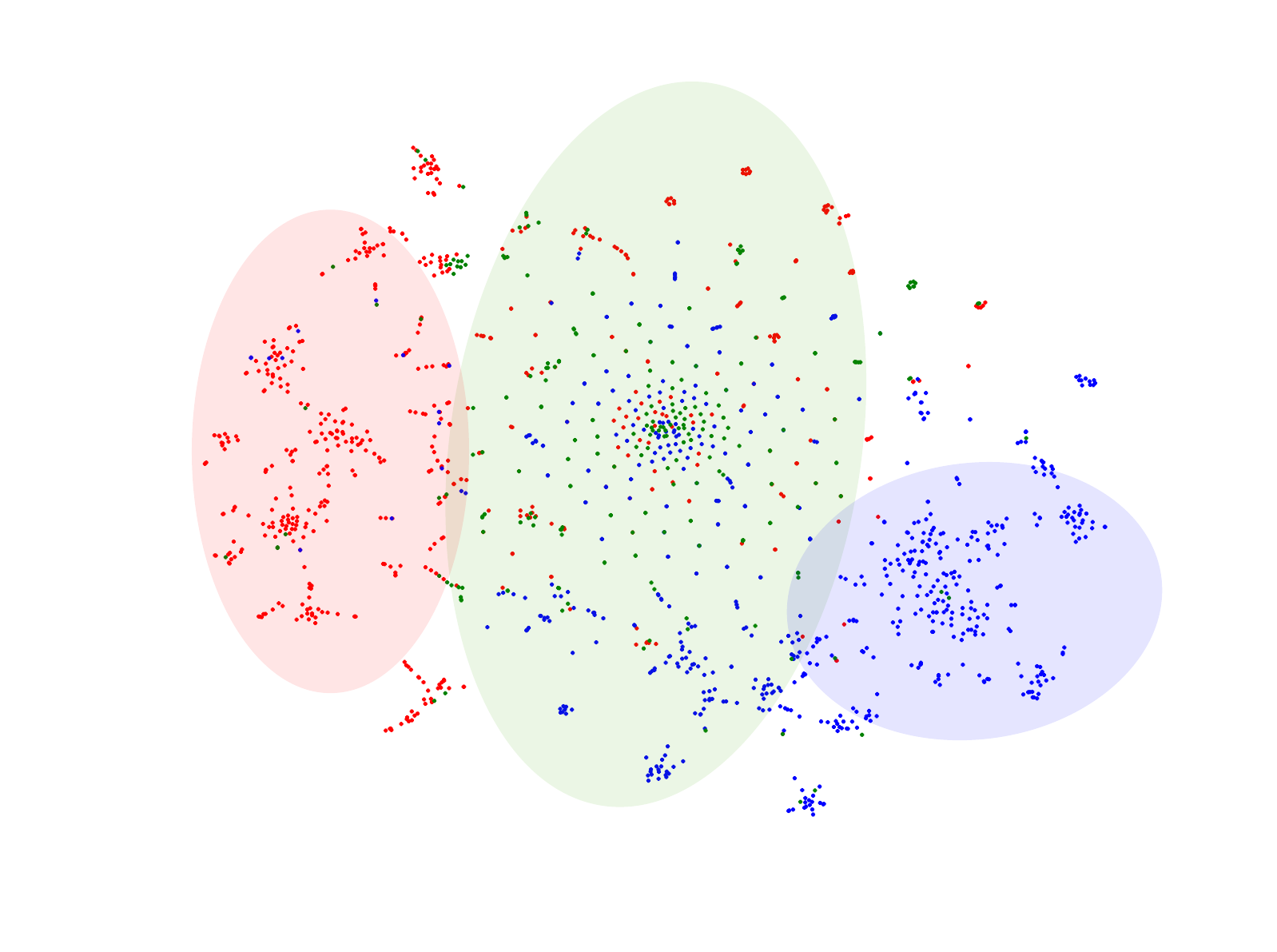}
      %\caption{fig2}
    \end{minipage}
  }%
  \centering
  \caption{Visualization of node clustering on Citeseer.}
  \label{Fig:NodeClusteringCiteseer}
\end{figure*}
\begin{figure*}[t]
  \centering
  \subfigure[\textbf{Original distribution}]{
    \begin{minipage}[t]{0.25\linewidth}
      \centering
      \label{Fig:NC_Cora}
      \includegraphics[width=\linewidth]{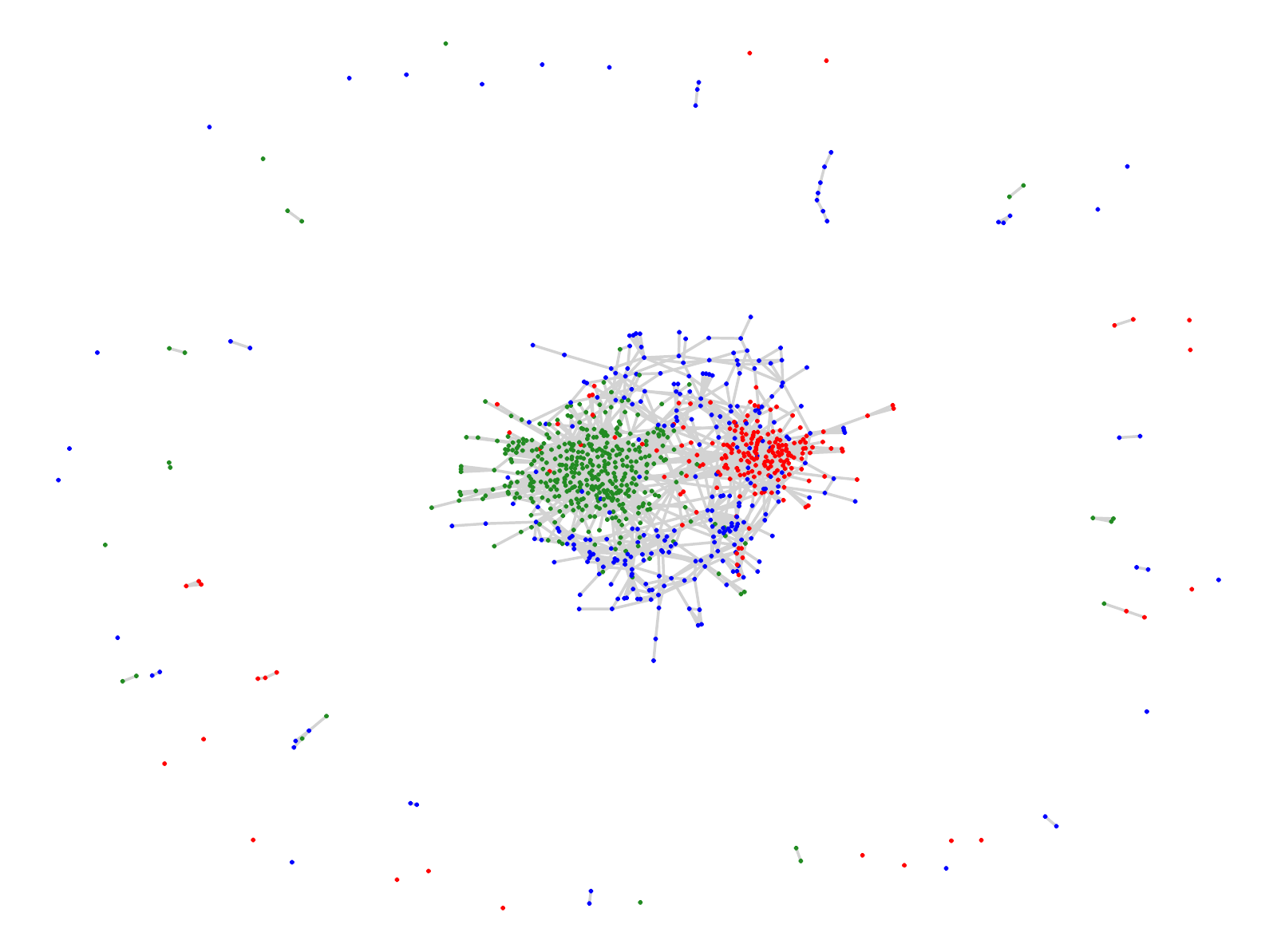}
      %\caption{fig1}
    \end{minipage}%
  }%
  \subfigure[\textbf{MLANE}]{
    \begin{minipage}[t]{0.25\linewidth}
      \centering
      \label{Fig:NC_Cora_MLANE}
      \includegraphics[width=\linewidth]{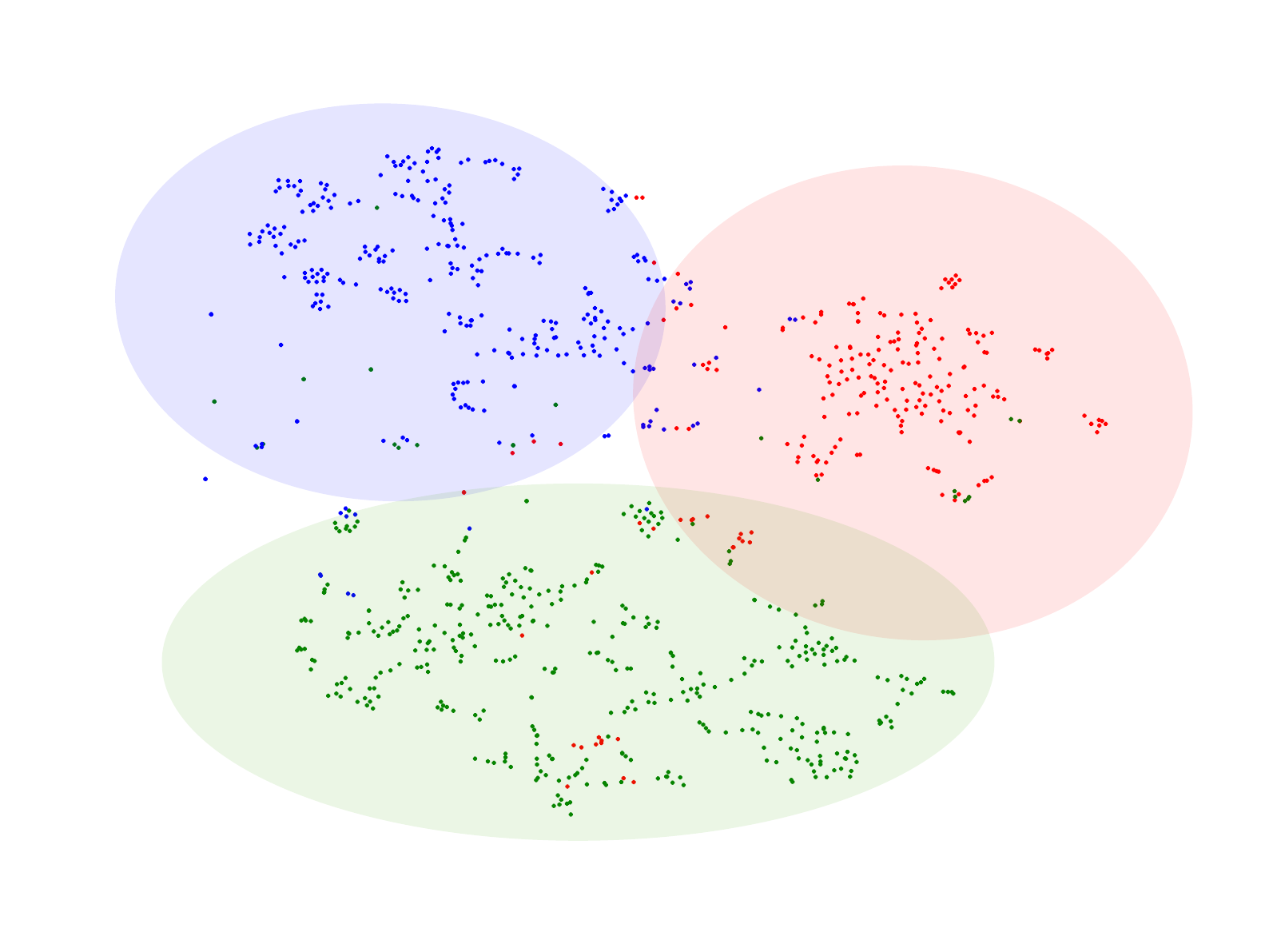}
      %\caption{fig2}
    \end{minipage}%
  }%
  \subfigure[\textbf{LINE}]{
    \begin{minipage}[t]{0.25\linewidth}
      \centering
      \label{Fig:NC_Cora_LINE}
      \includegraphics[width=\linewidth]{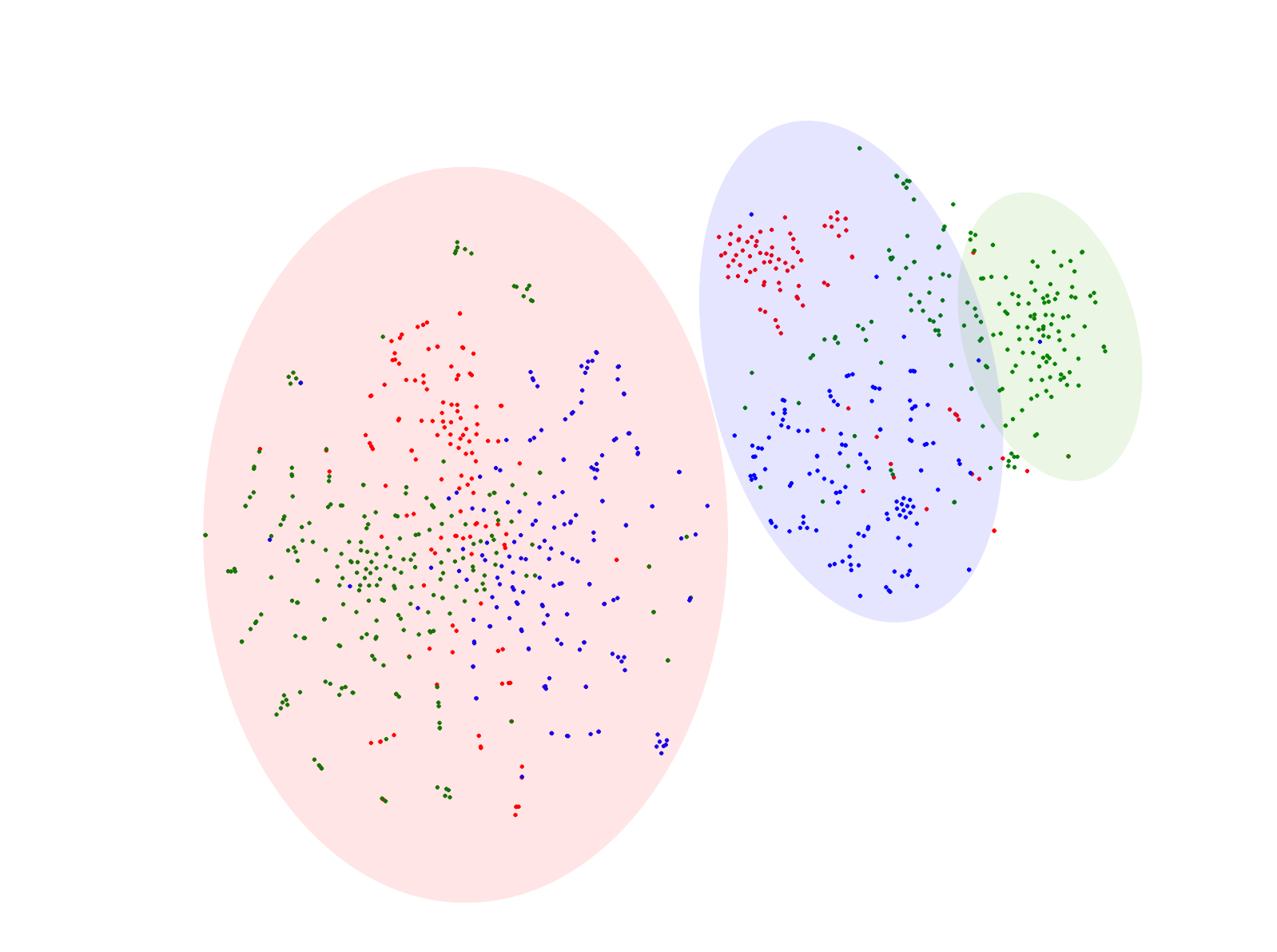}
      %\caption{fig2}
    \end{minipage}
  }%
  \subfigure[\textbf{node2vec}]{
    \begin{minipage}[t]{0.25\linewidth}
      \centering
      \label{Fig:NC_Cora_node2vec}
      \includegraphics[width=\linewidth]{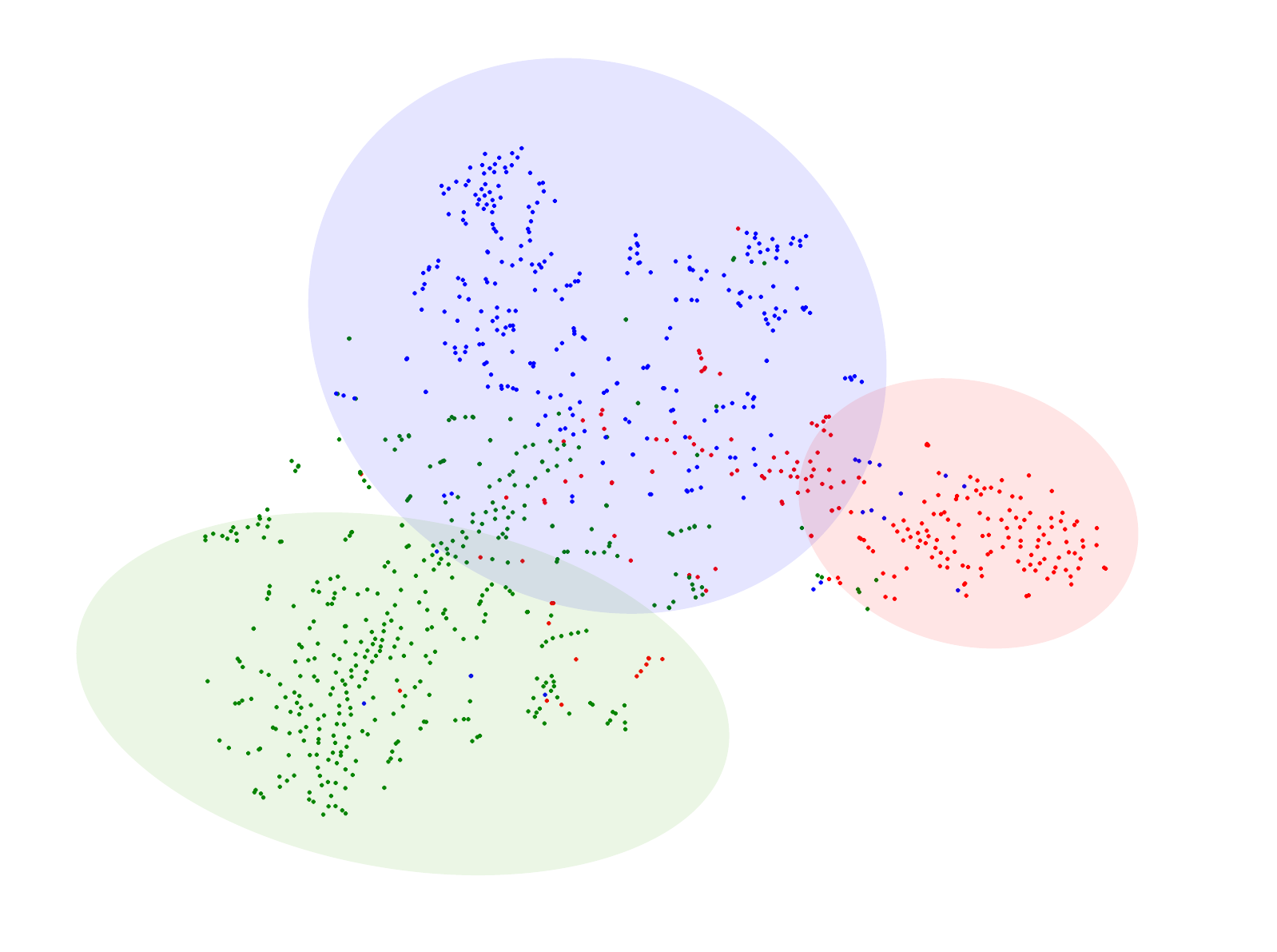}
      %\caption{fig2}
    \end{minipage}
  }%
  \centering
  \caption{Visualization of node clustering on Cora.}
  \label{Fig:NodeClusteringCora}
\end{figure*}
%%%%%%%%%%%%%%%%%%%%%%%%%%%%%%%%%%%%%%%%%%%%%%%%%%%%%%%%%%%%%%%%%%%%%%%%%%%%%%%%%%%%%%%%%%%%%%%%%%%%%%%%%%%%%%%%%%%%%%%%%%%%%%%%%%%%%

\subsection{Verification of Adaptiveness}
Now we further verify the adaptiveness of MLANE by a case study on a subgraph extracted from dataset PPI which contains three classes of gene nodes including EMT nodes, IR nodes, and Coagulation nodes. Figure~\ref{Fig:SampleVisualization} visualizes three sampling processes of MLANE for classification of the three nodes, where sampled nodes are marked with orange color and the visited edges are marked with blue color. The thicker an edge, the more frequent it is visited. We can obtain the following observations:

\begin{itemize}
\item Figure~\ref{Fig:SampleVisualizationIR} shows a sampling process of an IR node colored with purple. MLANE biases the sampling strategy for the IR node to DFS which tends to sample the nodes far from the IR node so that homophily can be preserved into the embedding of this node. 

\item Figure~\ref{Fig:SampleVisualizationEMT} shows a sampling process of an EMT node colored with red. MLANE biases its sampling strategy to BFS which tends to sample more neighbors so that the structural equivalence can be preserved into its embedding. 

\item At last, Figure~\ref{Fig:SampleVisualizationCoa} shows a sampling process of a Coagulation node colored with green. Similar to the sampling process of the IR node, MLANE also biases the the sampling strategy of this Coagulation node to DFS so that homophily can be preserved into the embedding of this node. 

\end{itemize}

In summary, this case study gives us a visual demonstration of the adaptiveness of MLANE, where the results are consistent with the intuition described in Section~\ref{Sec:Intro} that the embeddings of EMT nodes should preserve structural equivalence so that they can be correctly classified based on their embeddings while the embeddings of IR nodes and Coagulation nodes should preserve homophily in order to be classfied correctly based on their embeddings.

\subsection{Convergence Speed Analysis}
Now we investigate the convergence speed of MLANE. We run MLANE and baselines on randomly generated graphs with sizes from 1,000 to 10,000, where the node and label distributions are as same as Citeseer. For clear visualization, we only show the results of MLANE, AttentionWalk \cite{abu2018watch}, and GAT \cite{velivckovic2017graph} in Figure \ref{Fig:ConvergenceSpeed}, where the vertical axis represents the number of epochs when the gradient descent converges. %the convergence speed of MLANE, AttentionWalk, and GAT. Note that we divide the number of convergence epochs of GAT by 6 for clear comparison as the number of convergence epochs of GAT is much larger than that of MLANE and AttentionWalk. 
In Figure \ref{Fig:ConvergenceSpeed}, the numbers of epochs that MLANE and GAT need for convergence both grow approximately linear to the scale of datasets, while that of AttentionWalk grows exponentially. When the scale is greater than 8,000 nodes, the number of epochs of AttentionWalk becomes significantly larger than that of MLANE. Meanwhile, GAT always takes more epochs than MLANE. From Figure \ref{Fig:ConvergenceSpeed}, we can conclude that the convergence speed of MLANE is approximately linear to the scale of datasets, which suggests that MLANE can improve node embeddings without compromising the scalability.

%%%%%%%%%%%%%%%%%%%%%%%%%%%%%%%%%%%%%%%%%%%%%%%%%   Convergence speed figures  %%%%%%%%%%%%%%%%%%%%%%%%%%%%%%%%%%%%%%%%%%%%%%%%%
 \begin{figure}[t]
    \centering
    \includegraphics[width=0.75\linewidth]{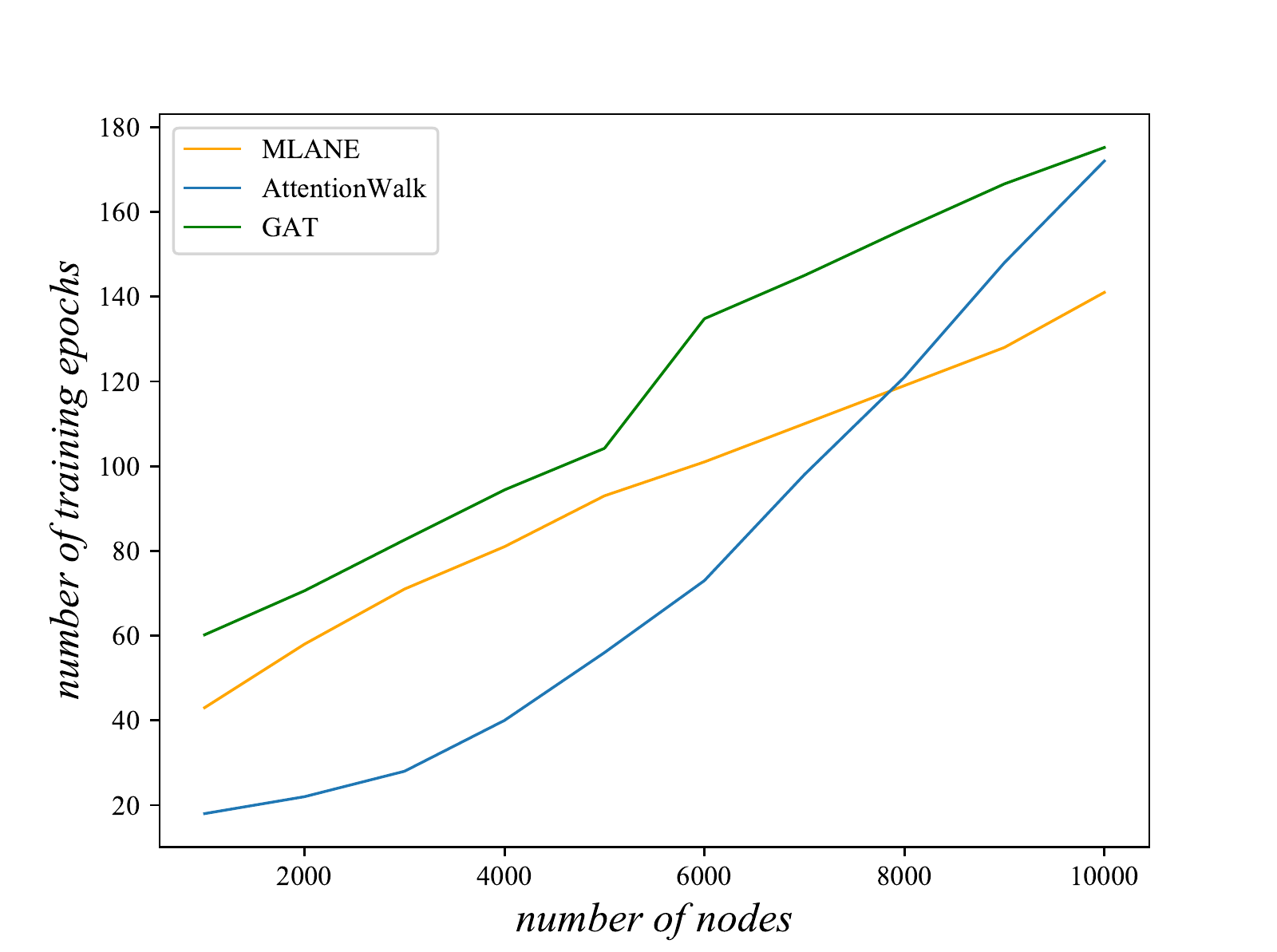}
    \caption{Convergence speed of MLANE, AttentionWalk, and GAT.}
    \label{Fig:ConvergenceSpeed}
  \end{figure}

\section{Related Work}
\label{Sec:RelatedWork}
\subsection{Network Embedding}
In terms of what kind of structural property is preserved, the existing network embedding methods can be roughly categorized into two classes. One class is to preserve \textit{homophily} of nodes \cite{mcpherson2001birds}, while the other class is to preserve \textit{structural equivalence} of nodes \cite{lorrain1971structural}. Homophily regularizes the learned embeddings with local connectivity so that the interconnected nodes have similar representations. For example, DeepWalk \cite{perozzi2014deepwalk}, LINE \cite{tang2015line}, AttentionWalk \cite{abu2018watch}, ProNE \cite{zhang2019prone}, SDNE \cite{wang2016structural}, HOPE \cite{ou2016asymmetric}, GAT \cite{velivckovic2017graph}, and DHPE \cite{Zhu8329541} learn node embeddings by preserving the first-order proximity or high-order proximity between nodes. Preserving homophily will benefit tasks like community detection, as nodes with similar label or features are more likely to be connected, but often fail in tasks like structure role identity \cite{Rossi6880836,ribeiro2017struc2vec}. In structural role identity, the nodes with similar local topology, even though without connection, play the same function and will be identified with the same role, where structural equivalence between nodes is the property desired to be preserved in the embedding learning. For example, by capturing structural equivalence between nodes, struct2vec \cite{ribeiro2017struc2vec}, RiWalk \cite{xuewma2019riwalk}, and DRNE \cite{tu2018deep} can generate similar embeddings for nodes of similar roles while dissimilar embeddings for nodes of different roles. As in real world nodes often exhibit a mixture of homophily and structural equivalence, Grover et al. propose a random walk based model called node2vec, which takes both properties into consideration via two hyperparameters ($p$ and $q$) controlling the probability of which one of the two search strategies, BFS (preferring to capture structure equivalence) and DFS (preferring to capture homophily), is chosen at each walk step \cite{grover2016node2vec}. However, all the existing random walk based methods separate the sampling process from embedding learning, which makes them nonadaptive and might degrade the embeddings due to the potentially inconsistent objectives of sampling and embedding.

\subsection{Meta-learning}
Meta-learning aims at learning to learn, which consists of a learner (model) responsible learning a task and a meta-learner (optimizer) responsible for learning how to train the learner \cite{vilalta2002perspective}. The existing meta-learning methods often focus on learning the optimizer that can be used for model training \cite{mishra2018simple}, or learning parameter initialization for fast adaptation \cite{finn2017model}. Recently, Peng et al. propose a meta-learner to learn undersampling strategy for class-imbalance learning \cite{peng2019trainable}. Different from the existing meta-learning methods, in this paper we propose a reinforcement learning based meta-learner to learn the search strategy for random walk based network embedding learning, by which the random walk based sampling process can be incorporated with embedding learning into an optimization problem that can be solved in an end-to-end fashion.

%\begin{figure}[t]
%  \centering
%  \subfigure[\textbf{embedding dimension}]{
%    \label{Fig:Parameter_d}
%    \begin{minipage}[t]{0.5\linewidth}
%      \centering
%      \includegraphics[width=\linewidth]{embed_dim.pdf}
%      %\caption{fig1}
%    \end{minipage}%
%  }%
%  \subfigure[\textbf{walk length}]{
%    \label{Fig:Parameter_L}
%    \begin{minipage}[t]{0.5\linewidth}
%      \centering
%      \includegraphics[width=\linewidth]{walk_length.pdf}
%      %\caption{fig2}
%    \end{minipage}%
%  }%
%  \quad
%  %这个回车键很重要 \quad也可以
%  \subfigure[\textbf{walks per node}]{
%    \label{Fig:Parameter_K}
%    \begin{minipage}[t]{0.5\linewidth}
%      \centering
%      \includegraphics[width=\linewidth]{walks_per_node.pdf}
%      %\caption{fig2}
%    \end{minipage}
%  }%
%  \subfigure[\textbf{window size}]{
%    \label{Fig:Parameter_w}
%    \begin{minipage}[t]{0.5\linewidth}
%      \centering
%      \includegraphics[width=\linewidth]{window_size.pdf}
%      %\caption{fig2}
%    \end{minipage}
%  }%B
%  \centering
%  \caption{Parameter sensitivity}
%  \label{Fig:ParameterSensitivity}
%\end{figure}

\section{Conclusion}
\label{Sec:Conclusion}
In this paper, we propose a novel method called Meta-Learning based Adaptive Network Embedding (MLANE). MLANE incorporates the random walk based sampling process with embedding learning into one optimization problem that can be solved via an end-to-end meta-learning framework based on reinforcement learning. By making the sampling process learnable, MLANE can adaptively preserve homophily and structural equivalence for different nodes in different tasks. Extensive experiments conducted on real datasets verify that due to the adaptiveness, MLANE can significantly improve the performance of node embeddings for down-stream network analysis tasks.

\section*{Acknowledgment}
This work is supported by National Natural Science Foundation of China under grant 61972270 and Hightech Program of Sichuan Province under grant 2019YFG0213.

\bibliographystyle{IEEEtran}
\bibliography{IEEEabrv,ANE_BigData}

\vspace{12pt}

\end{document}